\documentclass[lettersize,journal]{IEEEtran}
\usepackage{amsmath,amsfonts}
\usepackage{algorithmic}
\usepackage{algorithm}
\usepackage{array}
\usepackage{textcomp}
\usepackage{stfloats}
\usepackage{url}
\usepackage{verbatim}
\usepackage{graphicx}
\usepackage{cite}
\usepackage{hyperref}
\usepackage{amsthm}
\usepackage{amssymb}
\usepackage{subfigure}
\usepackage{color}
\usepackage{csquotes}
\usepackage{booktabs}
\usepackage{multirow}
\usepackage{pifont}
\usepackage{threeparttable}

\theoremstyle{definition}
\newtheorem{theorem}{Theorem} 

\theoremstyle{definition}
\newtheorem{lemma}{Lemma} 

\theoremstyle{definition}

\theoremstyle{definition}

\theoremstyle{definition}
\newtheorem{definition}{Definition}

\theoremstyle{definition}

\begin{document}

\title{An Accelerated Mixed Weighted-Unweighted MMSE Approach for MU-MIMO Beamforming}

\author{Xi Gao, Akang Wang, Junkai Zhang~\IEEEmembership{Member,~IEEE}, Qihong Duan, Jiang Xue~\IEEEmembership{Senior Member,~IEEE}
\thanks{Xi Gao, Junkai Zhang, Qihong Duan and Jiang Xue are with School of Mathematics and Statistics, Xi'an Jiaotong University, Xi'an, China (e-mail: gaoxi@stu.xjtu.edu.cn; jk.zhang@xjtu.edu.cn; duanqihong@xjtu.edu.cn; x.jiang@xjtu.edu.cn; corresponding authors: Jiang Xue and Akang Wang).}% <-this % stops a space
\thanks{Akang Wang is with Shenzhen International Center for Industrial and Applied Mathematics, Shenzhen Research Institute of Big Data, China and also with School of Data Science, The Chinese University of Hong Kong, Shenzhen, China (e-mail: wangakang@sribd.cn).}
}

\markboth{Journal of \LaTeX\ Class Files,~Vol.~14, No.~8, August~2021}%
{Shell \MakeLowercase{\textit{et al.}}: A Sample Article Using IEEEtran.cls for IEEE Journals}

\IEEEpubid{0000--0000/00\$00.00~\copyright~2021 IEEE}

\maketitle

\begin{abstract}
Precoding design based on weighted sum-rate (WSR) maximization is a fundamental problem in downlink multi-cell multi-user multiple-input multiple-output systems. 
While the weighted minimum mean-square error (WMMSE) algorithm is a standard solution, its high computational complexity—cubic in the number of base station antennas due to matrix inversions—hinders its application in latency-sensitive scenarios. 
To address this limitation, we propose a highly parallel algorithm based on a hybrid block coordinate descent framework. 
Our key innovation lies in updating the precoding matrix via block coordinate gradient descent, which avoids matrix inversions and relies solely on matrix multiplications, while the remaining variables are updated using block coordinate minimization. 
Furthermore, we introduce a two-stage warm-start strategy grounded in the sum mean-square error (MSE) minimization problem to accelerate convergence. 
We prove that the proposed algorithm converges to a stationary point of the WSR maximization problem. 
We refer to our method as the Accelerated Mixed Weighted-Unweighted Sum-MSE Minimization (A-MMMSE) algorithm. 
Simulation results demonstrate that A-MMMSE matches the WSR performance of both conventional WMMSE and its enhanced variant, reduced-WMMSE, while achieving a substantial reduction in computational time across diverse system configurations.

\end{abstract}


\section{Introduction}

\IEEEPARstart{M}{ulti-user} multiple-input multiple-output (MU-MIMO) systems serve as a foundational technology in both fifth-generation and emerging next-generation wireless networks~\cite{marzetta2016fundamentals, de2022overview}. 
A central challenge in MU-MIMO downlink systems lies in the design of transmit precoders, which aims primarily to mitigate inter-user interference and improve spectral efficiency. 
To achieve these objectives, prevailing strategies include sum mean-square error (MSE) minimization~\cite{hunger2006alternating,shi2007downlink} and weighted sum-rate~(WSR) maximization~\cite{shi2011iteratively, zhao2023rethinking, shen2018fractional1, shen2018fractional2} under a sum power constraint (SPC). 
Among them, WSR maximization offers a more direct pathway to achieving high spectral efficiency.
Nevertheless, the problem is inherently non-convex and $\mathcal{NP}$-hard~\cite{liu2010coordinated}, which motivates the pursuit of efficient computational algorithms.

Although obtaining a globally optimal solution for the WSR maximization problem is desirable, the exponential complexity of global optimization methods~\cite{joshi2011weighted, liu2012achieving} makes suboptimal precoding strategies more practical for MU-MIMO systems. 
These include zero-forcing (ZF)~\cite{gao2011linear, nguyen2019multi}, successive convex approximation (SCA)~\cite{venturino2008successive, kaleva2012weighted}, and the weighted minimum mean-square error (WMMSE) algorithm~\cite{christensen2008weighted, shi2011iteratively}. 
Among them, ZF-based approaches are recognized as computationally efficient but deliver relatively lower sum rates. 
In contrast, SCA-based techniques, which rely on constructing a sequence of surrogate functions, incur higher computational load and exhibit slower convergence. 
The WMMSE algorithm generally strikes a balance, achieving higher sum rates than ZF and faster convergence than SCA by leveraging simple closed-form updates.

The WMMSE algorithm solves the WSR maximization problem by reformulating it into an equivalent \textit{block multi-convex} problem, which is convex in each block of variables when the others are fixed~\cite{xu2013block}.
Using the \textit{block coordinate descent}~(BCD) method~\cite{beck2013convergence,xu2013block}, this approach is guaranteed to converge to a stationary point of the original WSR maximization problem. 
Another related methodology is \textit{fractional programming}. 
This technique reformulates the non-convex WSR maximization into a block multi-convex form via \textit{quadratic transformation} (QT)~\cite{shen2018fractional1} or \textit{Lagrangian dual transformation}~\cite{shen2018fractional2}, thereby enabling the application of BCD for efficient solution.
It has been shown that WMMSE can be derived as a special case of the latter.
However, these methods exhibit rapidly increasing computational complexity in MU-MIMO systems as the number of base station~(BS) antennas grows. 
{
Specifically, the precoding update requires a matrix inversion with $\mathcal{O}(M^3)$ complexity, where $M$ denotes the number of BS antennas. 
The bisection search introduced to ensure optimality and feasibility further increases the runtime of WMMSE, as the precoding matrix must be repeatedly updated within the search loop.
}

\IEEEpubidadjcol

\begin{table*}[h!]
	\centering
	\caption{{Comparison of existing gradient-based beamforming methods.}}
	\label{tab:frt_method_comp}
        \resizebox{0.9\linewidth}{!}{
		\begin{tabular}{l|cccccccc}
		\toprule
		{Methods} & {Applicable Scenarios} & {$\mathbf{U}$} {Inversion} & {$\mathbf{W}$ Inversion} & {$\mathbf{V}$ Inversion}  \\
        \midrule
		{PGD-WMMSE~\cite{pellaco2021matrix}} & {MU-MISO} & {-} & {-} & {\ding{55}} \\
        {Matrix-Inverse-Free WMMSE~\cite{pellaco2023matrix}} & {MU-MIMO} & {\ding{55}} & {\ding{55}} & {\ding{55}}  \\
        {Nonhomogeneous Quadratic Transform (NQT)~\cite{shen2024accelerating}} & {Multi-Cell MU-MIMO} &  {\checkmark}  & {-} & {\ding{55}}  \\
        {Nonhomogeneous Fractional Programming~\cite{chen2025fast}} & {Multi-Cell ISAC} & {\checkmark}  &  {-} & {\ding{55}}  \\
        {\textbf{A-MMMSE}} & {Multi-Cell MU-MIMO} &{\checkmark}  & {\ding{55} $\to$ \checkmark}  & {\ding{55}}\\
		\bottomrule
	\end{tabular}}
    % tablenotes
\end{table*}

To address the aforementioned challenges, an \textit{algorithm unrolling} framework known as matrix-inverse-free WMMSE was introduced in \cite{pellaco2021matrix, pellaco2023matrix}. 
The core idea is to eliminate matrix inversions using first-order methods to improve algorithmic parallelism. 
The work in \cite{pellaco2021matrix} first applied multi-step \textit{projected gradient descent}~(PGD) to replace the matrix inversion in the precoder update for multi-user multiple-input single-output systems (MU-MISO). 
This direction was further extended in \cite{pellaco2023matrix}, which incorporated the Schulz iteration~\cite{schulz1933iterative} to avoid all matrix inversions, thereby generalizing the matrix-inverse-free WMMSE approach to MU-MIMO scenarios. 
While these unrolling methods achieved better empirical performances, they often suffer from limited generalization capability.
Within MU-MIMO systems, the matrix inversion during the precoder update remains the most computationally expensive operation in WMMSE, while inversion operations for auxiliary variables, with the complexity proportional to the number of user antennas, impose only negligible computational overhead.
To avoid the precoder matrix inversion operations, the authors of~\cite{zhao2023rethinking} demonstrated that any non-trivial solution to WSR maximization must satisfy the SPC with equality. 
This key finding reveals an inherent low-dimensional subspace structure within the optimal precoder. 
Leveraging this structure, they proposed a reduced-complexity algorithm named reduced-WMMSE~(R-WMMSE).
{
Subsequent work~\cite{zhao2025universal} extended R-WMMSE from the single-cell WSR problem to a general beamforming optimization problem in multi-cell systems. 
While this structural insight reduces the computational complexity to linear in the number of BS antennas, it introduces cubic complexity in the number of users. 
In particular, the computational complexity for multi-cell scenarios becomes cubic in the total number of users across all cells.
} 
This cubic dependence adversely impacts the efficiency of R-WMMSE and presents a scalability limitation in large-scale user scenarios.
For QT~\cite{shen2018fractional1}, the authors of~\cite{shen2024accelerating} proposed to optimize the non-homogeneous bound of the quadratic term in the transformed objective function, thereby circumventing the need for matrix inversion.
Subsequently, this method was extended to multi-cell beamforming design for integrated sensing and communications (ISAC)~\cite{chen2025fast}. 
They also incorporated Nesterov's extrapolation to accelerate convergence, which has been empirically shown to improve the performance of WMMSE when applied specifically to the precoder update step~\cite{zhou2023novel}. 
However, this algorithm is only applicable to scenarios with a single data stream.

Concurrently, an emerging trend focuses on the development of first-order optimization methods, capitalizing on their inherent parallelism and distributed computing capabilities.
These methods rely solely on gradient information and offer notable computational advantages over second-order approaches.
Algorithms such as the primal–dual hybrid gradient method~\cite{he2014convergence} and the alternating direction method of multipliers~\cite{boyd2011distributed} have successfully achieved accuracy comparable to conventional solvers in domains including linear programming~\cite{applegate2021practical}, convex quadratic programming~\cite{schubiger2020gpu, lu2025practical}, and conic programming~\cite{lin2025pdcs}.
In parallel, \textit{warm-start} strategies remain a fundamental acceleration technique, in which solutions from previously solved, closely related problems are used to initialize new optimization instances.
This practice often leads to significant reductions in the number of iterations required for convergence. 
In contrast, when no prior information is available, the process starts from a \textit{cold-start} initialization.
Warm-starting is widely adopted in sequential programming~\cite{morelli2024warm} and across algorithms like gradient descent~\cite{tao2023laws, sambharya2024learning} and interior-point methods~\cite{john2008implementation, benson2008interior, gao2024ipm}.
Nevertheless, most existing methods for WSR maximization still rely on cold-start initialization, overlooking potential gains from systematic initialization.

Inspired by these advances, we revisit the precoding update in WMMSE from a BCD perspective. 
{We propose to leverage first-order techniques to reduce the computational cost of solving this subproblem and to accelerate convergence, while still ensuring convergence to a stationary point of the original WSR problem.} 
Furthermore, we introduce a warm-start strategy based on sum-MSE minimization to accelerate the overall convergence of WMMSE.
The proposed method is termed \textit{accelerated mixed weighted-unweighted sum-MSE minimization} (A-MMMSE).
{The key differences between our approach and other first-order method-based precoding algorithms are summarized in Table~\ref{tab:frt_method_comp}.}
The main contributions of this work are summarized as follows:
\begin{itemize}
    \item We introduce a hybrid BCD strategy for WMMSE wherein the computationally expensive precoding matrices are optimized using \textit{block coordinate gradient descent} (BCGD), while low-complexity block variables are updated via \textit{block coordinate minimization} (BCM). An extrapolation technique is further incorporated to accelerate convergence.

    \item We develop a two-stage warm-start strategy that initially solves a sum-MSE minimization problem and then switches to the weighted sum-MSE formulation once the relative change in the WSR falls below a predefined threshold. 
    This approach can be employed to accelerate any WMMSE-based method.

    \item We establish the convergence of A-MMMSE, proving that this method guarantees convergence to a stationary point of the WSR maximization problem.

    \item {Simulation results confirm that A-MMMSE achieves almost the same final WSR performance as WMMSE and R-WMMSE. 
    In large-scale BS antenna scenarios, A-MMMSE achieves a CPU runtime comparable to R-WMMSE and significantly faster than WMMSE. 
    R-MMMSE, which augments R-WMMSE with the mixed weighted-unweighted warm-start method, achieves the lowest CPU runtime and is 40.8\% faster than R-WMMSE. 
    In large-scale user scenarios, A-MMMSE reduces CPU runtime by up to 6× and 5× compared to WMMSE and R-WMMSE, respectively. 
    When deployed on GPUs, A-MMMSE leverages its parallelism and outperforms R-WMMSE in large-scale BS antenna scenarios.}
\end{itemize}

The remainder of this paper is organized as follows. Section~\ref{sec:sys_mod} formulates the multi-cell MU-MIMO system model and the WSR maximization problem under the SPC. 
Section~\ref{sec:bcd_method} introduces the conventional BCD framework. 
Section~\ref{sec:wmmse} reviews the classical WMMSE algorithm. 
Section~\ref{sec:AMMSE_approach} presents the proposed hybrid BCD scheme and the mixed weighted-unweighted warm-start strategy. 
Section~\ref{sec:converg} provides a convergence analysis. 
Section~\ref{sec:experiments} validates the effectiveness of the proposed approach through numerical experiments. 
Finally, Section~\ref{sec:concl} concludes the paper.

\textbf{Notation}: Scalars, vectors, and matrices are denoted by lowercase, boldface lowercase, and boldface uppercase letters, respectively. 
The space of $ M \times N $ complex matrices is denoted by $ \mathbb{C}^{M \times N} $. 
For a matrix $ \mathbf{A} $, $ \mathbf{A}^{H} $, $ \mathbf{A}^{-1} $, $\operatorname{Tr}(\mathbf{A}) $, $\sigma_i(\mathbf{A})$ and $ \lambda_i(\mathbf{A}) $ represent its Hermitian transpose, inverse, trace, $i$-th singular value, and $i$-th eigenvalue, respectively. 
The spectral and Frobenius norms are denoted by $\|\mathbf{A}\|_2$ and $\|\mathbf{A}\|_F$. 
Finally, $ \mathcal{CN}(\mathbf{0}, \mathbf{\Sigma}) $ denotes a circularly symmetric complex Gaussian distribution with zero mean and covariance matrix $ \mathbf{\Sigma}$.

\section{System Model and Problem Formulation}
\label{sec:sys_mod}

\begin{figure*}
    \centering
\includegraphics[width=0.9\linewidth]{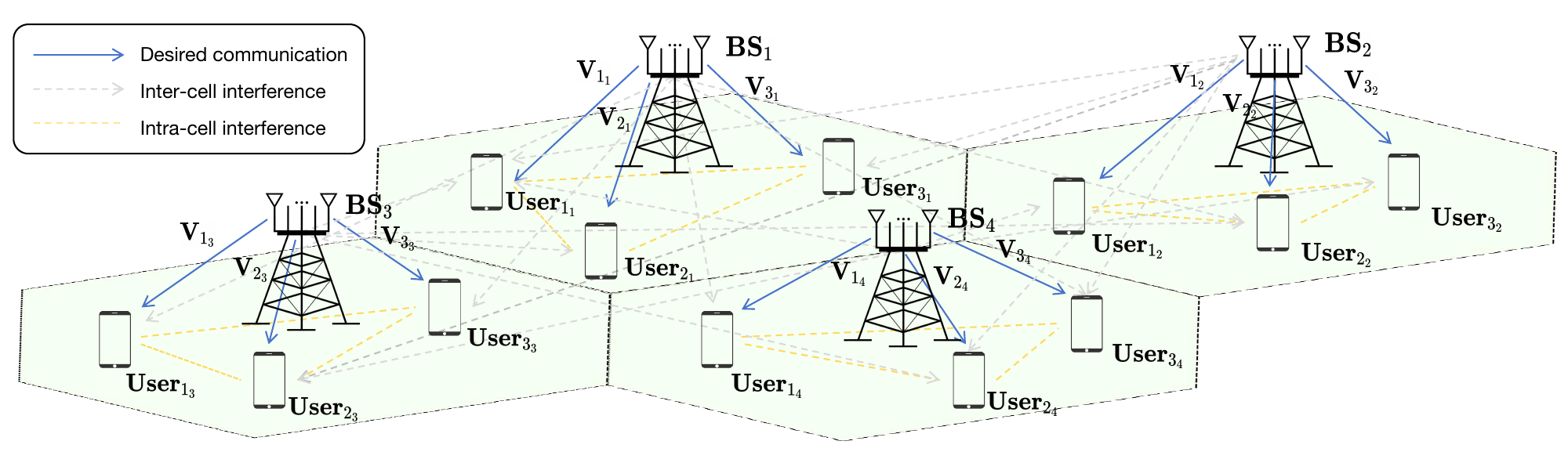}
    \caption{{An illustration of downlink multi-cell MU-MIMO precoding for WSR maximization.}}
    \label{fig:sys_model}
\end{figure*}

{
We consider a downlink MU-MIMO system consisting of $K$ cells, as illustrated in Fig.~\ref{fig:sys_model}. 
In the $k$-th cell, the BS is equipped with $M$ transmit antennas and serves $I$ users. 
Let $i_k$ denote the $i$-th user in the $k$-th cell. 
Each user is equipped with $N$ receive antennas and is capable of receiving $d$ independent data streams simultaneously.
Let $\mathbf{s}_{i_k} \in \mathbb{C}^{d \times 1}$ denote the symbol vector intended for user $i_k$, and let $\mathbf{V}_{i_k} \in \mathbb{C}^{M \times d}$ represent the corresponding linear precoder at the BS. 
The received signal at user $i_k$ is then given by:
\begin{equation*}
\mathbf{y}_{i_k} = \mathbf{H}_{i_k k} \mathbf{V}_{i_k} \mathbf{s}_{i_k} + \sum_{(l,j) \neq (i,k)} \mathbf{H}_{i_k j} \mathbf{V}_{lj} \mathbf{s}_{lj} + \mathbf{n}_{i_k}, \quad \forall i_k \in \mathcal{I},
\end{equation*}
where $\mathbf{H}_{i_k j} \in \mathbb{C}^{N \times M}$ denotes the channel matrix from the $j$-th BS to user $i_k$, and $\mathbf{n}_{i_k} \in \mathbb{C}^{N \times 1}$ represents the additive white Gaussian noise vector distributed as $\mathcal{CN}(\mathbf{0}, \sigma_{i_k}^2 \mathbf{I})$. 
The set $\mathcal{I}$ collects all users across the entire system and is defined as:
\begin{equation*}
\mathcal{I} = \left\{ i_k \mid k \in \{1, 2, \dots, K\}, \; i \in \{1, 2, \dots, I\} \right\}.
\end{equation*}
}
{
A fundamental challenge lies in designing the optimal set of precoders $\mathbf{V} \triangleq \{\mathbf{V}_{i_k}\}_{i_k \in \mathcal{I}}$ that maximizes the WSR subject to the SPCs. 
This problem is formulated as follows:
\begin{subequations}
\label{prob:ori_prob}
\begin{align}
    \max _{\mathbf{V}} & \sum_{k=1}^{K} \sum_{i=1}^{I} \alpha_{i_k} R_{i_k} \label{prob:ori_prob_obj} \\
    \text{s.t.} & \sum_{i=1}^{I} \operatorname{Tr}\left(\mathbf{V}_{i_k} \mathbf{V}_{i_k}^{H}\right) \leq P_{k}, \quad \forall k = 1, 2, \dots, K, \label{prob:ori_prob_constr}
\end{align}
\end{subequations}
where \(\alpha_{i_k}\) denotes the priority weight of user \(i_k\), and
\begin{equation*}
\begin{aligned}
R_{i_k} \triangleq \log \operatorname{det} & \left( \mathbf{I} + \mathbf{H}_{i_k k} \mathbf{V}_{i_k} \mathbf{V}_{i_k}^{H} \mathbf{H}_{i_k k}^{H} \right. \\
& \left. \left( \sum_{(l,j) \neq (i,k)} \mathbf{H}_{i_k j} \mathbf{V}_{l_j} \mathbf{V}_{l_j}^{H} \mathbf{H}_{i_k j}^{H} + \sigma_{i_k}^{2} \mathbf{I} \right)^{-1} \right)
\end{aligned}
\end{equation*}
represents the achievable rate of user \(i_k\). Here, \(P_{k}\) denotes the total transmit power budget at the \(k\)-th BS. For notational simplicity, we refer to the \(k\)-th SPC in \eqref{prob:ori_prob_constr} as \(\mathcal{P}_{k}\).
}

\IEEEpubidadjcol

\section{BCD for Block Multi-Convex Optimization}
\label{sec:bcd_method}

Problem~\eqref{prob:ori_prob} is non-convex with respect to the precoder $\mathbf{V}$. 
A standard approach to tackle this issue involves reformulating the problem into a block multi-convex structure~\cite{xu2013block}, typically achieved through either equivalent transformations~\cite{shen2018fractional1,shen2018fractional2} or the introduction of auxiliary variables~\cite{shi2011iteratively}.
While the reformulated problem may remain jointly non-convex in all variables, its convexity with respect to each individual block variable makes it amenable to the BCD framework. 
Within this framework, each subproblem can be solved efficiently and optimally. 

To illustrate the BCD framework, we consider the following problem:
\begin{equation}
\label{prob:block_prob}
\min_{\mathbf{X} \in \mathcal{X}} f\left(\mathbf{X}_{1}, \ldots, \mathbf{X}_{s}\right)
\end{equation}
where $\mathbf{X} \triangleq \{\mathbf{X}_1, \ldots, \mathbf{X}_s\}$ is partitioned into $s$ blocks, $\mathcal{X}$ is a closed block multi-convex set, and $f$ is a differentiable block multi-convex function. 
Both $\mathcal{X}$ and $f$ may be nonconvex jointly over $\mathbf{X}$.

A set $\mathcal{X}$ is block multi-convex if its projection onto each block is convex. 
Formally, for each $i$ and fixed $\mathbf{X}_1, \ldots, \mathbf{X}_{i-1}, \mathbf{X}_{i+1}, \ldots, \mathbf{X}_s$, the set
\begin{equation*}
\begin{aligned}
&\mathcal{X}_{i}\left(\mathbf{X}_{1}, \ldots, \mathbf{X}_{i-1}, \mathbf{X}_{i+1}, \ldots, \mathbf{X}_{s}\right) \\
\triangleq &
\left\{\mathbf{X}_{i} \in \mathbb{C}^{m_i\times n_{i}}:\left(\mathbf{X}_{1}, \ldots, \mathbf{X}_{i-1}, \mathbf{X}_{i}, \mathbf{X}_{i+1}, \ldots, \mathbf{X}_{s}\right) \in \mathcal{X}\right\}
\end{aligned}
\end{equation*}
is convex. 
Similarly, $f$ is block multi-convex if it is convex in each block variable $\mathbf{X}_i$ with all other blocks held fixed. 
Thus, fixing all but one block reduces problem~\eqref{prob:block_prob} to a convex subproblem.

We focus on the Gauss-Seidel BCD variant, which cyclically minimizes $f$ over each block while keeping others fixed at their latest values. 
Let $\mathbf{X}_i^t$ denote $\mathbf{X}_i$ after its $t$-th update, and define:
\begin{equation*}
\begin{aligned}
    f_{i}^{t}\left(\mathbf{X}_{i}\right) \triangleq f\left( \right.& \mathbf{X}_{1}^{t}, \ldots, \mathbf{X}_{i-1}^{t}, \mathbf{X}_{i},  
     \left.\mathbf{X}_{i+1}^{t-1}, \ldots, \mathbf{X}_{s}^{t-1}\right),
\end{aligned}
\end{equation*}
\begin{equation*}
\begin{aligned}
   \mathcal{X}_{i}^{t}\triangleq 
\mathcal{X}_{i}\left(\mathbf{X}_{1}^{t}, \ldots, \mathbf{X}_{i-1}^{t}, \mathbf{X}_{i+1}^{t-1}, \ldots, \mathbf{X}_{s}^{t-1}\right).
\end{aligned}
\end{equation*}
At each step, we consider two different updates:
\begin{equation}
    \label{eq:bcd_origi}
   \hspace{-2.7cm} \mathbf{X}_{i}^{t}\leftarrow \underset{\mathbf{X}_{i} \in \mathcal{X}_{i}^{t}}{\operatorname{argmin}} \; f_{i}^{t}\left(\mathbf{X}_{i}\right),
\end{equation}
\begin{equation}
\label{eq:bcd_proxi_linear}
\begin{aligned}
\mathbf{X}_{i}^{t}\leftarrow&\underset{\mathbf{X}_{i} \in \mathcal{X}_{i}^{t}}{\operatorname{argmin}}\left\langle\nabla f_{i}^{t}\left(\mathbf{X}_{i}^{t-1}\right), \mathbf{X}_{i}-\mathbf{X}_{i}^{t-1}\right\rangle\\
&+\frac{L_{i}^{t-1}}{2}\left\|\mathbf{X}_{i}-\mathbf{X}_{i}^{t-1}\right\|_2^{2},
\end{aligned}
\end{equation}
where $L_i^{t-1} > 0$ is a penalty parameter, and $\nabla f_i^t(\mathbf{X}_i^{t-1})$ is the block-partial gradient.
The complete BCD framework is summarized in Algorithm~\ref{alg:bcd}. 
The algorithm consists of two loops: an outer loop that controls the overall iteration count to ensure convergence, and an inner loop that cyclically updates all block variables.
Each $\mathbf{X}_i$ may be updated using either~\eqref{eq:bcd_origi} or~\eqref{eq:bcd_proxi_linear}, with the requirement that the selected update rule remains consistent across all iterations. 

Update rule~\eqref{eq:bcd_origi} corresponds to the conventional BCD approach, commonly referred to as BCM. 
In contrast, update rule~\eqref{eq:bcd_proxi_linear} minimizes the first-order Taylor expansion of $ f_i^t $ at $ \mathbf{X}_i^{t-1} $, closely resembling BCGD. 
To ensure convergence of Algorithm~\ref{alg:bcd}, $\nabla f_i^t(\mathbf{X}_i)$ is required to be Lipschitz continuous for every block $i$ to which~\eqref{eq:bcd_proxi_linear} is applied.
The parameter $L^t_i$ in~\eqref{eq:bcd_proxi_linear} may be held constant or adjusted adaptively during iterations.
The computational cost of solving~\eqref{eq:bcd_proxi_linear} is significantly lower than that of~\eqref{eq:bcd_origi}, making it particularly advantageous when exact minimization~\eqref{eq:bcd_origi} is computationally expensive. 
To further accelerate convergence, extrapolation is commonly employed as shown in~(\ref{eq:bcd_extrap}):
\begin{equation}
\label{eq:bcd_extrap}
    \hat{\mathbf{X}}_{i}^{t-1} = \mathbf{X}_{i}^{t-1} + \omega_{i}^{t-1}\left(\mathbf{X}_{i}^{t-1} - \mathbf{X}_{i}^{t-2}\right),
\end{equation}
where $0 \leq \omega_{i}^{t-1} < 1$ is the extrapolation coefficient. 
Thus, $\mathbf{X}_{i}^{t-1}$ in~\eqref{eq:bcd_proxi_linear} is replaced by $\hat{\mathbf{X}}_{i}^{t-1}$.

% The computational cost of solving~\eqref{eq:bcd_proxi_linear} is significantly lower, making it particularly advantageous when updates via~\eqref{eq:bcd_origi} are computationally expensive.
% \blue{To further improve the convergence, a common technique is
% %The update process in~\eqref{eq:bcd_proxi_linear} can be substantially accelerated through 
% to utilize extrapolation, as shown in~(\ref{eq:bcd_extrap}):
% \begin{equation}
% \label{eq:bcd_extrap}
%     \hat{\mathbf{X}}_{i}^{t-1} = \mathbf{X}_{i}^{t-1} + \omega_{i}^{t-1}\left(\mathbf{X}_{i}^{t-1} - \mathbf{X}_{i}^{t-2}\right),
% \end{equation}
% where $0\leq \omega_{i}^{t-1} \leq 1$ denotes the extrapolation coefficient, and $\mathbf{X}_{i}^{t-1}$ in~\eqref{eq:bcd_proxi_linear} is replaced by $\hat{\mathbf{X}}_{i}^{t-1}$.}

\begin{algorithm}
\caption{BCD for solving~\eqref{prob:block_prob}.}
\begin{algorithmic}
\STATE \textbf{Initialize: } $\left(\mathbf{X}_{1}^{-1}, \ldots, \mathbf{X}_{s}^{-1}\right)=\left(\mathbf{X}_{1}^{0}, \ldots, \mathbf{X}_{s}^{0}\right)$, $t=0$.
\REPEAT
\FOR{$i=1,2,\ldots,s$}
\STATE $\mathbf{X}_i^t \leftarrow \eqref{eq:bcd_origi}$ or \eqref{eq:bcd_proxi_linear};
\ENDFOR
\STATE $t \gets t+1$;
\UNTIL{stopping criterion is satisfied.}
\end{algorithmic}
\label{alg:bcd}
\end{algorithm}

\section{The classical WMMSE}
\label{sec:wmmse} 
The WMMSE algorithm is a well-established method for solving problem~\eqref{prob:ori_prob}. 
It transforms the original problem into an equivalent block multi-convex structure and solves it using BCD methods. 
The problem addressed by WMMSE is closely related to the following sum-MSE minimization formulation:
{
\begin{subequations}
\label{prob:mmse_prob}
\begin{align}
\min _{\mathbf{U}, \mathbf{V}} & \sum_{k=1}^{K}\sum_{i=1}^{I}\alpha_{i_k}\operatorname{Tr}\left(\mathbf{E}_{i_k}\right) \\
\text { s.t. } & \sum_{i=1}^{I} \operatorname{Tr}\left(\mathbf{V}_{i_k} \mathbf{V}_{i_k}^{H}\right) \leq P_{k},\quad \forall k =1,2,\dots,K
\end{align}
\end{subequations}
}
where 
{
\begin{equation}
\label{eq:Ek}
\begin{aligned}
\mathbf{E}_{i_k} \triangleq & \left(\mathbf{I}-\mathbf{U}_{i_k}^{H} \mathbf{H}_{i_kk} \mathbf{V}_{i_k}\right)\left(\mathbf{I}-\mathbf{U}_{i_k}^{H} \mathbf{H}_{i_kk} \mathbf{V}_{i_k}\right)^{H} \\
& + \sum_{(l,j) \neq (i,k)} \mathbf{U}_{i_k}^{H}\mathbf{H}_{i_kj} \mathbf{V}_{l_j} \mathbf{V}_{l_j}^{H} \mathbf{H}_{i_kj}^{H} \mathbf{U}_{i_k}+\sigma_{i_k}^{2}\mathbf{U}_{i_k}^{H} \mathbf{U}_{i_k},
\end{aligned}
\end{equation}
}and {$\mathbf{U} \triangleq\left\{\mathbf{U}_{i_k}\in\mathbb{C}^{N\times d}\right\}_{i_k\in \mathcal{I}}$} denotes the set of receive precoders.
By introducing auxiliary variables {$\mathbf{W} \triangleq\{\mathbf{W}_{i_k}\in\mathbb{C}^{d\times d}\}_{i_k\in \mathcal{I}}$}, WMMSE constructs a matrix-weighted sum-MSE minimization problem that is equivalent to the original WSR maximization~\eqref{prob:ori_prob}:
{
\begin{subequations}
\label{prob:wmmse_prob}
\begin{align}
 \min _{\mathbf{U}, \mathbf{W}, \mathbf{V}} & \sum_{k=1}^{K}\sum_{i=1}^{I} \alpha_{i_k}\left(\operatorname{Tr}\left(\mathbf{W}_{i_k} \mathbf{E}_{i_k}\right)-\log \operatorname{det}\left(\mathbf{W}_{i_k}\right)\right) \label{prob:wmmse_prob_obj}\\
 \quad \text{s.t.} & \sum_{i=1}^{I} \operatorname{Tr}\left(\mathbf{V}_{i_k} \mathbf{V}_{i_k}^{H}\right) \leq P_{k},\quad \forall k=1,2,\cdots,K. \label{prob:wmmse_prob_constr}
\end{align}
\end{subequations}
}
Problem~\eqref{prob:wmmse_prob} exhibits convexity with respect to each individual variable $\mathbf{U}$, $\mathbf{W}$, and $\mathbf{V}$ separately, {provided that $\mathbf{W}$ is positive definite}. 
Initializing with a feasible point $\mathbf{V}^{0}$, the WMMSE algorithm iteratively updates all block variables through the exact minimization procedure specified in~\eqref{eq:bcd_origi}.
Specifically, for every $i_k$, the update of $\mathbf{U}$ while fixing $\mathbf{W}$ and $\mathbf{V}$ is given by
{
\begin{equation}
\label{eq:Uk}
\mathbf{U}_{i_k}=\left(\sum_{j=1}^{K}\sum_{l=1}^{I} \mathbf{H}_{i_kj} \mathbf{V}_{l_j} \mathbf{V}_{l_j}^{H} \mathbf{H}_{i_kj}^{H}+\sigma_{i_k}^{2} \mathbf{I}\right)^{-1} \mathbf{H}_{i_k} \mathbf{V}_{i_k},
\end{equation}
}and the update of $\mathbf{W}$ while fixing the other two block variables is given by
{\begin{equation}
\label{eq:Wk}
\mathbf{W}_{i_k}=\left(\mathbf{I}-\mathbf{U}_{i_k}^{H} \mathbf{H}_{i_kk} \mathbf{V}_{i_k}\right)^{-1}.
\end{equation}
The definition of $\mathbf{U}_{i_k}$ ensures that $\mathbf{W}_{i_k}$ naturally remains positive definite.
}

While fixing $\mathbf{U}$ and $\mathbf{W}$, the precoder $\mathbf{V}_{i_k}$ is obtained by solving the following convex quadratically constrained quadratic program:
{
\begin{small}
    \begin{equation}
\label{prob:P_subproblem}
\begin{aligned}
\min_{\{\mathbf{V}_{i_k}\}_{i=1}^{I}} \quad &\sum_{i=1}^{I} \alpha_{i_k} \operatorname{Tr}\left(\mathbf{W}_{i_k}\left(\mathbf{I}-\mathbf{U}_{i_k}^{H} \mathbf{H}_{i_k k} \mathbf{V}_{i_k}\right)\right.\\
&\left.\left(\mathbf{I}-\mathbf{U}_{i_k}^{H} \mathbf{H}_{i_k k} \mathbf{V}_{i_k}\right)^{H}\right) \\
& + \sum_{i=1}^{I} \sum_{(l,j)\neq (i,k)} \alpha_{l_j} \operatorname{Tr}\left(\mathbf{W}_{l_j} \mathbf{U}_{l_j}^{H} \mathbf{H}_{l_j k} \mathbf{V}_{i_k} \mathbf{V}_{i_k}^{H} \mathbf{H}_{l_j k}^{H} \mathbf{U}_{l_j}\right) \\
\text{s.t.} \quad & \sum_{i=1}^{I} \operatorname{Tr}\left(\mathbf{V}_{i_k} \mathbf{V}_{i_k}^{H}\right) \leq P_{k}.
\end{aligned}
\end{equation}
\end{small}
}
The update rule for $\mathbf{V}$ can be derived from the optimality conditions of the Lagrangian function corresponding to problem~\eqref{prob:P_subproblem}, as shown in~(\ref{eq:wmmse_Vk}):
{
\begin{equation}
\label{eq:wmmse_Vk}
    \begin{aligned}
        \mathbf{V}_{i_k} = &\alpha_{i_k} \left(\sum_{j=1}^K\sum_{l=1}^{I}\alpha_{l_j}\mathbf{H}_{l_jk}^{H}\mathbf{U}_{l_j}\mathbf{W}_{l_j}\mathbf{U}_{l_j}^{H}\mathbf{H}_{l_jk}+\lambda_k\mathbf{I}\right)^{-1}\\
    &\mathbf{H}_{i_kk}^{H}\mathbf{U}_{i_k}\mathbf{W}_{i_k}, \quad \forall i_k \in \mathcal{I},
    \end{aligned}
\end{equation}
}where $\lambda_k \geq 0$ is the dual variable associated with the constraint $\mathcal{P}_k$ and can be determined via a bisection method.

The complete WMMSE algorithm, based on the update steps above, is summarized in Algorithm~\ref{alg:wmmse}. 
While the updates for $\mathbf{U}$ and $\mathbf{W}$ involve matrix inversions, these operations incur only minor computational overhead in practice, since mobile communication systems typically constrain $d \leq N \leq 4$.
In contrast, the update for $\mathbf{V}$ presents a significant computational bottleneck. 
The matrix inversion required in this step exhibits $\mathcal{O}(M^3)$ complexity, which becomes prohibitive as the number of BS antennas $M$ grows large in MU-MIMO systems.
This issue is further exacerbated by the bisection search for the dual variable $\lambda_k$, which requires repeated computation of the $\mathbf{V}$ subproblem until the sum power constraints is satisfied.

\begin{algorithm}
\caption{The WMMSE Framework.}
\begin{algorithmic}
\STATE \textbf{Input: }The tolerance of accuracy $\epsilon$.
\STATE \textbf{Output: } $\mathbf{V}$.
\STATE \textbf{Initialize: } {$\mathbf{V}^{0}_{i_k}\in \mathcal{P}_{k},\forall i_k\in\mathcal{I}$}, $t=1$.
\REPEAT 
\STATE Update {$\mathbf{U}_{i_k}$} by~\eqref{eq:Uk};
\STATE Update {$\mathbf{W}_{i_k}$} by~\eqref{eq:Wk};
\STATE Update {$\mathbf{V}_{i_k}$} by~\eqref{eq:wmmse_Vk};
\STATE $t=t+1$;
\UNTIL{{{$\left|\left(\sum\limits_{(i,k)} \alpha_{i_k} R_{i_k}^{t} - \sum\limits_{(i,k)} \alpha_{i_k} R_{i_k}^{t-1}\right) \Big/ \sum\limits_{(i,k)} \alpha_{i_k} R_{i_k}^{t-1}\right| \leq \epsilon.$}}}
\end{algorithmic}
\label{alg:wmmse}
\end{algorithm}

\section{The A-MMMSE APPROACH}  \label{sec:AMMSE_approach}
\subsection{Accelerated Hybrid BCD Techniques}
\label{sec:accelerating_tech}
The overall complexity of an algorithm depends on its per-iteration cost and the number of iterations to convergence. While WMMSE converges in fewer iterations due to its BCM updates, this comes at a high per-iteration cost. We therefore propose to trade a modest increase in iteration count for a drastic reduction in per-iteration complexity, leading to faster overall computation.

As discussed in Section~\ref{sec:wmmse}, the primary computational bottleneck in WMMSE is the update of the precoding matrix $\mathbf{V}$~\eqref{eq:wmmse_Vk}. 
To overcome this, we integrate the efficient BCGD scheme from~\eqref{eq:bcd_proxi_linear} into the BCD framework, replacing the more costly BCM update in~\eqref{eq:bcd_origi}. To preserve global convergence and a strong convergence rate, we retain the standard BCM updates for $\mathbf{U}$ and $\mathbf{W}$. This results in a hybrid BCD algorithm that employs mixed update strategies.
Specifically, we define the objective function in~\eqref{prob:wmmse_prob_obj} as {$f(\mathbf{U}, \mathbf{W}, \mathbf{V}) \triangleq \sum_{k=1}^{K}\sum_{i=1}^{I} \alpha_{i_k}\left(\operatorname{Tr}\left(\mathbf{W}_{i_k} \mathbf{E}_{i_k}\right)-\log \operatorname{det}\left(\mathbf{W}_{i_k}\right)\right)$}. 
The update of {$\mathbf{V}_{i_k}$} via~\eqref{eq:bcd_proxi_linear} is equivalent to a PGD method, formulated as:
{
\begin{equation}
\label{eq:bcgd_Vk}
\mathbf{V}_{i_k}^{t} = \Pi_{\mathcal{P}_{k}}\big(\mathbf{V}^{t-1}_{i_k} - \gamma \nabla_v f^t(\mathbf{V}_{i_k}^{t-1})\big),\quad \forall i_k\in\mathcal{I},
\end{equation}
}where $\gamma$ denotes the step size, {$\nabla_v f^t(\mathbf{V}_{i_k}^{t-1})$} represents the gradient with respect to {$\mathbf{V}_{i_k}^{t-1}$} at $t$-th iteration, given by:
{
\begin{equation}
\begin{aligned}
    \nabla_v f^t(\mathbf{V}_{i_k}^{t-1}) & \triangleq \left(\sum_{(l,j)} 2\alpha_{l_j}\mathbf{H}_{l_jk}^H\mathbf{U}_{l_j}^t\mathbf{W}_{l_j}^t\mathbf{U}_{l_j}^{t}{}^{H}\mathbf{H}_{l_jk}\right)\mathbf{V}_{i_k}^{t-1} \\
    & - 2\alpha_{i_k}\mathbf{H}_{i_kk}^H\mathbf{U}_{i_k}^t\mathbf{W}_{i_k}^t,\\
\end{aligned}
\end{equation}}and $\Pi_{\mathcal{P}_{k}}(\mathbf{V})$ denotes the projection operator onto the $k$-th SPC:
{
\begin{equation}
\label{eq:projec_method}
\Pi_{\mathcal{P}_{k}}(\mathbf{V}) \triangleq
\begin{cases}
\mathbf{V} & \text{if } \mathbf{V} \in \mathcal{P}_{k}, \\
\mathbf{V} \cdot \sqrt{\frac{P_{k}}{\sum_{i=1}^{I}\operatorname{Tr}(\mathbf{V}_{i_k}\mathbf{V}_{i_k}^{H})}} & \text{if } \mathbf{V} \notin \mathcal{P}_{k}.
\end{cases}
\end{equation}
The update rule in~\eqref{eq:bcgd_Vk} involves only matrix multiplications, reducing the computational complexity for $\mathbf{V}_k$ from $\mathcal{O}(IKMNd + IKMd^2 + IKM^2d + M^3)$ in~\eqref{eq:wmmse_Vk} to $\mathcal{O}(IKMNd + IKMd^2 + IKM^2d)$ in~\eqref{eq:bcgd_Vk}, thereby eliminating the $\mathcal{O}(M^3)$ complexity associated with matrix inversion.
In contrast, R-WMMSE~\cite{zhao2025universal} also removes the $\mathcal{O}(M^3)$ complexity via its low-dimensional subspace approach. 
Nevertheless, it suffers from $\mathcal{O}((IKN)^3)$ complexity in the total number of users, whereas the update rule in~\eqref{eq:bcgd_Vk} retains linear complexity in the same aspect.

The choice of step size $\gamma$ significantly affects the convergence rate of first-order methods. 
Selecting the optimal step size $\gamma_{i_k}^t = 1 / L_{v_{i_k}}^t$ for $\mathbf{V}_{i_k}^t$, where $L_{v_{i_k}}^t$ denotes the Lipschitz constant of $\nabla_v f^t(\mathbf{V}_{i_k}^{t-1})$, still involves matrix decomposition with $\mathcal{O}(M^3)$ complexity. 
Nevertheless, thanks to the projection operation~\eqref{eq:projec_method} that eliminates the need for bisection search, the computational cost still remains lower than that of WMMSE. 
On the other hand, fixing $\gamma$ to a constant reduces computational cost but often leads to a slower convergence rate.
We observe that the BCGD method typically employs a relatively large step size in the early iterations, which gradually decreases when approaching a stationary point. 
Accordingly, we adopt the following step size selection strategy:
\begin{equation}
\label{eq:gamma}
\gamma^t_{i_k} \triangleq
\begin{cases}
1 / \| \sum_{(l,j)} 2\alpha_{l_j} \mathbf{H}_{l_j k}^H \mathbf{U}_{l_j}^t \mathbf{W}_{l_j}^t \mathbf{U}_{l_j}^{t}{}^{H} \mathbf{H}_{l_j k} \|_2 & \text{if } t \leq T_{c}, \\
\gamma_c & \text{otherwise}.
\end{cases}
\end{equation}
Here, $T_{c}$ is a very small number of iterations (e.g., $T_c = 3$ in our experiments), and $\gamma_c$ is a fixed step size. 
This step size design implies that our algorithm computes the optimal step size only in the initial iterations to ensure fast convergence and then switches to a fixed step size after $T_c$ iterations to reduce computational cost. This approach strikes a favorable balance between convergence speed and computational efficiency.
}

As discussed in Section~\ref{sec:bcd_method}, extrapolation techniques serve as a common acceleration strategy for the update scheme in~\eqref{eq:bcd_proxi_linear}.
{These techniques incorporate historical information from previous iterations by forming an extrapolated point through a weighted combination of the current point and past differences, thereby avoiding oscillatory behavior and accelerating convergence. 
Such extrapolation approaches are widely used in first-order methods~\cite{xu2013block} and have also been applied directly to WMMSE to achieve acceleration~\cite{zhou2023novel}. 
Furthermore, all the gradient-based methods listed in Table~\ref{tab:frt_method_comp} employ Nesterov's extrapolation to accelerate convergence.} 
Therefore, we incorporate an additional extrapolation step for $\mathbf{V}$ in each iteration:
{
\begin{equation}
\label{eq:extro_v}
\hat{\mathbf{V}}^{t-1}_{i_k} = \mathbf{V}^{t-1}_{i_k}+\omega(\mathbf{V}^{t-1}_{i_k}-\mathbf{V}^{t-2}_{i_k})\quad\forall i_k\in\mathcal{I},
\end{equation}
}where $0\leq\omega <1$ denotes the extrapolation parameter and {is computed as
\begin{equation}
    \label{eq:omega}
    \omega^t = \max\left\{\frac{t-2}{t+1}, 0\right\}, \quad \forall t \geq 1,
\end{equation}
following the setting in~\cite{shen2024accelerating}.}

{
Combined with the extrapolation technique in~\eqref{eq:extro_v}, we obtain an accelerated method based on hybrid BCD. 
This BCD-based acceleration technique is easy to parallelize and shares conceptual similarities with the matrix-free WMMSE approach~\cite{pellaco2021matrix, pellaco2023matrix}. 
However, unlike methods that eliminate all matrix inversions, our method specifically applies the BCGD method to the precoding update in MU-MIMO systems.
Furthermore, due to the inherently parallel nature of first-order methods, our approach can be executed not only on CPUs but also directly leverage the computational advantages of GPUs without requiring neural network-based algorithm unrolling. 
}

\subsection{Mixed Weighted-Unweighted sum-MSE Optimization}
In this section, we propose a heuristic warm-start strategy to accelerate the convergence of Algorithm~\ref{alg:wmmse}. 
An effective warm-start strategy provides a high-quality initial point for the main iterative process at minimal computational cost, thereby reducing overall computation time. 
{We first observe that the sum-MSE minimization problem~\eqref{prob:mmse_prob} is equivalent to the following problem:
\begin{subequations}
\label{prob:mmse_equi_prob}
\begin{align}
 \min _{\mathbf{U}, \mathbf{W}, \mathbf{V}} & \sum_{k=1}^{K}\sum_{i=1}^{I} \alpha_{i_k}\left(\operatorname{Tr}\left(\mathbf{W}_{i_k} \mathbf{E}_{i_k}\right)-\log \operatorname{det}\left(\mathbf{W}_{i_k}\right)\right) \label{prob:mmmse_prob_obj}\\
 \quad \text{s.t.} & \sum_{i=1}^{I} \operatorname{Tr}\left(\mathbf{V}_{i_k} \mathbf{V}_{i_k}^{H}\right) \leq P_{k},\quad \forall k=1,2,\cdots,K, \label{prob:mmmse_prob_constr}\\
&\mathbf{W}_{i_k}=\mathbf{I},\quad \forall i_k\in \mathcal{I} 
\end{align}
\end{subequations}
Compared with the weighted sum-MSE minimization problem~\eqref{prob:wmmse_prob}, this problem introduces an additional affine constraint $\mathbf{W}_{i_k} = \mathbf{I}$ for all $i_k \in \mathcal{I}$. Consequently, any solution obtained by solving the sum-MSE minimization problem~\eqref{prob:mmse_prob} is guaranteed to be feasible for problem~\eqref{prob:wmmse_prob}. 
Furthermore, although the convergence rate of BCD for general nonconvex problems remains an open problem~\cite{lyu2020convergence}, for convex problems, the linear convergence rate upper bound depends at least linearly on the number of block variables~\cite{sun2015improved}.
Since the sum-MSE minimization problem~\eqref{prob:mmse_prob} has fewer block variables (specifically, no $\mathbf{W}$) than the weighted sum-MSE minimization problem~\eqref{prob:wmmse_prob}, and furthermore, the structural similarity between the two formulations leads to comparable iterative updates and convergence paths, we heuristically propose solving~\eqref{prob:mmse_prob} to provide a feasible initial solution for~\eqref{prob:wmmse_prob}, thereby accelerating the overall convergence.}
When applying BCD to problem~\eqref{prob:mmse_prob}, the update for $\mathbf{U}$ remains unchanged from~\eqref{eq:Uk}, while the precoding matrix $\mathbf{V}$ is updated as:
{
\begin{equation}
\label{eq:mmse_Vk}
    \begin{aligned}
        \mathbf{V}_{i_k} = &\alpha_{i_k} \left(\sum_{j=1}^K\sum_{l=1}^{I}\alpha_{l_j}\mathbf{H}_{l_jk}^{H}\mathbf{U}_{l_j}\mathbf{U}_{l_j}^{H}\mathbf{H}_{l_jk}+\mu_k\mathbf{I}\right)^{-1}\\
        &\mathbf{H}_{i_kk}^{H}\mathbf{U}_{i_k}, \quad \forall i_k \in \mathcal{I},
    \end{aligned}
\end{equation}
}where ${\mu_k \geq 0}$ denotes the dual variable, obtained via bisection search.

%We address problem~\eqref{prob:mmse_prob} to accelerate convergence while ensuring the final solution constitutes a stationary point of problem~\eqref{prob:ori_prob}. 
We aim to accelerate the convergence of WMMSE for solving problem~\eqref{prob:mmse_prob} while ensuring the final solution constitutes a stationary point of problem~\eqref{prob:ori_prob}. 
To this end, we propose a two-stage method. The first stage solves problem~\eqref{prob:mmse_prob} until the relative change in WSR falls below threshold $\epsilon_1$. The algorithm then transitions to the second stage, which executes standard WMMSE updates until the relative change in WSR drops below $\epsilon_2$. 
Since the computational difference between solving~\eqref{prob:mmse_prob} and~\eqref{prob:wmmse_prob} lies solely in the update of $\mathbf{W}$, the per-iteration complexity remains comparable to that of original WMMSE. 
{Furthermore, this approach preserves the positive definiteness of $\mathbf{W}$ throughout the iterations and thus does not affect the block multi-convex property of the original problem.}
This approach effectively uses the MMSE solution to warm-start WMMSE. 
We term this method a \textit{mixed weighted-unweighted sum MSE minimization}~(MMMSE) algorithm, as illustrated in Fig.~\ref{fig:structure}.
Moreover, this modification could not only accelerate the conventional WMMSE algorithm but also enhance the computational efficiency of a range of related improved methods, such as R-WMMSE~\cite{zhao2023rethinking}. 
We will empirically demonstrate this performance improvement in Section~\ref{sec:experiments}.

\begin{figure}[h!]
    \centering
    \includegraphics[width=1.0\linewidth]{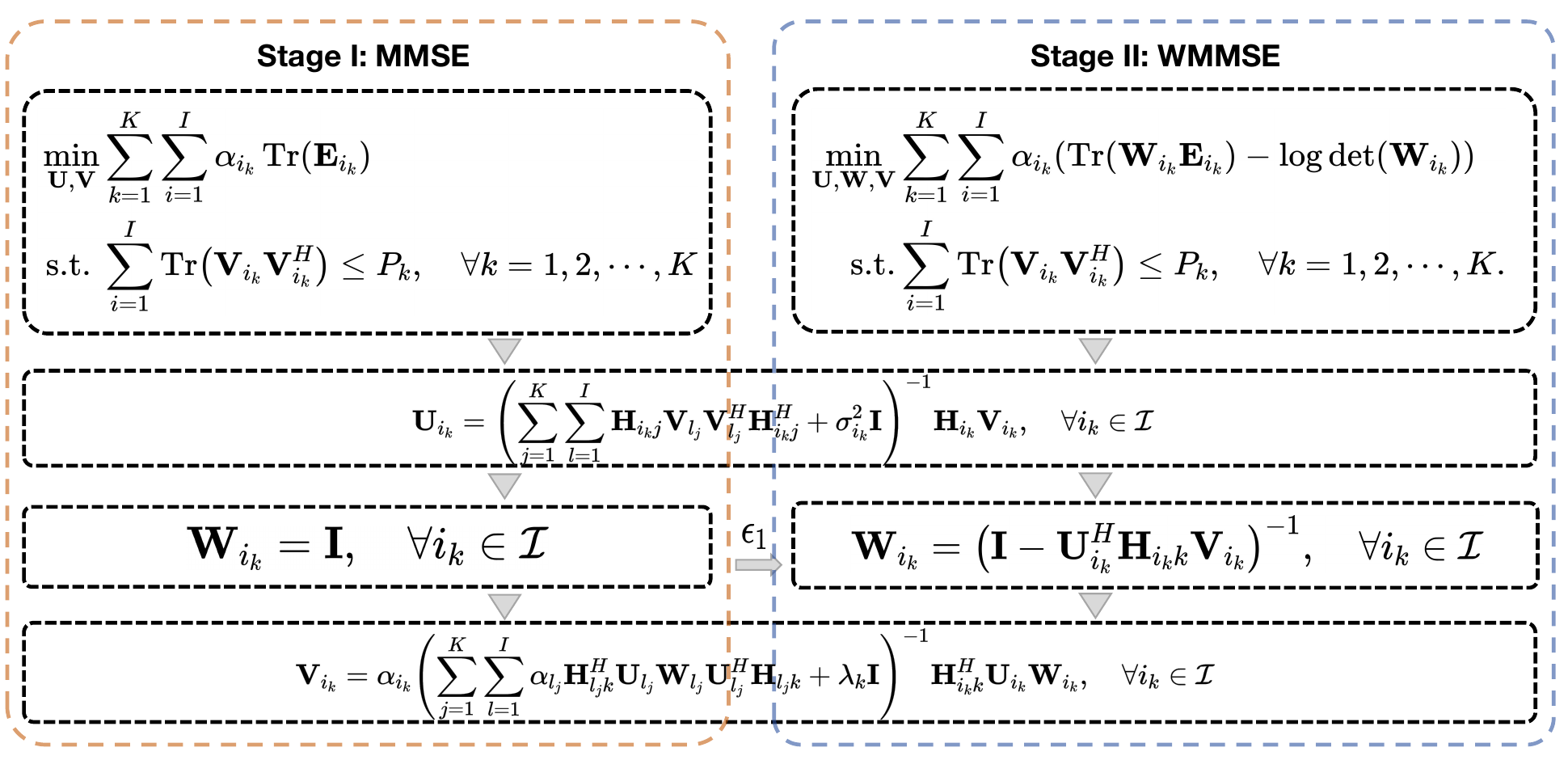}
    \caption{{Illustration of the proposed MMMSE algorithm.}}
    \label{fig:structure}
\end{figure}

Combined with the hybrid BCD acceleration technique, the proposed method is denoted as A-MMMSE and summarized in Algorithm~\ref{alg:inexact_wmmse}.
Given an initial feasible solution $\mathbf{V}^0$, A-MMMSE sequentially updates $\mathbf{U}$, $\mathbf{W}$, and $\mathbf{V}$, with each variable update amenable to distributed computation across users. 
For $t > 2$, an extrapolation technique is incorporated, allowing $\mathbf{V}$ updates to utilize historical information. 
During initial iterations, all {$\mathbf{W}_{i_k}$} are fixed as the identity matrix $\mathbf{I}$; once the solution accuracy satisfies $\epsilon_1$, the standard WMMSE update rule for $\mathbf{W}$ is restored, thereby reducing computational complexity for both $\mathbf{W}$ and subsequent $\mathbf{V}$ updates.
The variable $\mathbf{V}$ is updated via BCGD with projection to maintain feasibility, which significantly reduces the $\mathcal{O}(M^3)$ complexity associated with matrix inversion and eliminates the bisection search previously required to ensure the satisfaction of constraints~\eqref{prob:wmmse_prob_constr}.
Finally, the algorithm terminates when the relative change in the weighted sum rate falls below $\epsilon_2$.

\begin{algorithm}
\caption{A-MMMSE.}
\begin{algorithmic}
\STATE \textbf{Input: }The step size {$\gamma_c$}, the tolerance $\epsilon_1, \epsilon_2$.
\STATE \textbf{Output: }{$\mathbf{V}_{i_k}, \forall i_k\in\mathcal{I}$}.
\STATE \textbf{Initialize: } {$\mathbf{V}^{0}_{i_k}\in \mathcal{P}_{k},\forall i_k\in\mathcal{I}$,} $t=1$.
\REPEAT
\STATE {Compute $\omega^t$ using~\eqref{eq:omega};}
\IF{$t\geq 2$}
\STATE {$\hat{\mathbf{V}}_{i_k}^{t-1}=\mathbf{V}_{i_k}^{t-1}+\omega^t\left(\mathbf{V}_{i_k}^{t-1}-\mathbf{V}^{t-2}_{i_k}\right)$;}
\ELSE
\STATE {$\hat{\mathbf{V}}_{i_k}^{t-1}=\mathbf{V}_{i_k}^{t-1}$;}
\ENDIF
\STATE 
{Update $\mathbf{U}^{t}_{i_k}$ for all $i_k\in \mathcal{I}$ by~\eqref{eq:Uk};}
\IF{\small{{$\left|\left(\sum\limits_{(i,k)} \alpha_{i_k} R_{i_k}^{t} - \sum\limits_{(i,k)} \alpha_{i_k} R_{i_k}^{t-1}\right) \Big/ \sum\limits_{(i,k)} \alpha_{i_k} R_{i_k}^{t-1}\right| \leq \epsilon_1$}}}
\STATE 
{Update $\mathbf{W}^{t}_{i_k}$ for all $i_k\in \mathcal{I}$ by~\eqref{eq:Wk};}
\ELSE
\STATE {$\mathbf{W}_{i_k}^{t}=\mathbf{I}, \forall i_k\in\mathcal{I}$;}
\ENDIF
\STATE {Compute $\gamma_{i_k}^t$ for all $i_k\in \mathcal{I}$ using~\eqref{eq:gamma};}
\STATE {Update $\mathbf{V}^{t}_{i_k}$ for all $i_k\in \mathcal{I}$ by~\eqref{eq:bcgd_Vk};}
\STATE $t = t+1$;
\UNTIL{\small{{$\left|\left(\sum\limits_{(i,k)} \alpha_{i_k} R_{i_k}^{t} - \sum\limits_{(i,k)} \alpha_{i_k} R_{i_k}^{t-1}\right) \Big/ \sum\limits_{(i,k)} \alpha_{i_k} R_{i_k}^{t-1}\right| \leq \epsilon_2.$}}}
\end{algorithmic}
\label{alg:inexact_wmmse}
\end{algorithm}

\section{Convergence Analysis}
\label{sec:converg}
This section establishes the convergence of A-MMMSE for sufficiently small step sizes $\gamma > 0$. Since the first stage of A-MMMSE serves only to provide an initial point—satisfying the SPC via projection—to the BCD procedure in the second stage, we restrict our convergence analysis exclusively to this second stage. We begin by formalizing the matrix inner product and stationary point concepts that form the foundation of our convergence framework.
\begin{definition}[Complex Matrix Inner Product~\cite{pellaco2023matrix}]
    \label{def:inner_product}
    The inner product between complex matrices $\mathbf{X}$ and $\mathbf{Y}$, denoted $\langle \mathbf{X}, \mathbf{Y} \rangle$, is defined as the Euclidean inner product of their real-valued vector representations, obtained by concatenating the real and imaginary parts of the transposed rows.
\end{definition}

\begin{definition}[Stationary Point~\cite{xu2013block}]
    \label{def:nash_point}
    Let \(\overline{\mathbf{X}} \triangleq (\overline{\mathbf{X}}_1, \ldots, \overline{\mathbf{X}}_s)\). 
    A point \(\overline{\mathbf{X}}\) is a stationary point of \(\min_{\mathbf{X} \in \mathcal{X}} f(\overline{\mathbf{X}})\) if it is an interior point of $\mathcal{X}$ and satisfies:
    \begin{small}
        \begin{equation*}
        \label{cond:nash_cond_1}
        \begin{split}
            &f\left(\overline{\mathbf{X}}_{1}, \ldots, \overline{\mathbf{X}}_{i-1}, \overline{\mathbf{X}}_{i}, \overline{\mathbf{X}}_{i+1}, \ldots, \overline{\mathbf{X}}_{s}\right) \leq \\ 
        &f\left(\overline{\mathbf{X}}_{1}, \ldots, \overline{\mathbf{X}}_{i-1}, \mathbf{X}_{i}, \overline{\mathbf{X}}_{i+1}, \ldots, \overline{\mathbf{X}}_{s}\right), \quad \forall \mathbf{X}_{i} \in \overline{\mathcal{X}}_{i},\forall i
        \end{split}
    \end{equation*}
    \end{small}
    or equivalently,
    \begin{equation*}
        \label{cond:nash_cond_2}
        \left\langle\nabla_{\mathbf{X}_{i}} f(\overline{\mathbf{X}}), \mathbf{X}_{i}-\overline{\mathbf{X}}_{i}\right\rangle \geq 0, \quad \forall \mathbf{X}_{i} \in \overline{\mathcal{X}}_{i},\forall i,
    \end{equation*}
    where $\overline{\mathcal{X}}_{i}\triangleq 
\mathcal{X}_{i}\left(\overline{\mathbf{X}}_{1}, \ldots, \overline{\mathbf{X}}_{i-1}, \overline{\mathbf{X}}_{i+1}, \ldots, \overline{\mathbf{X}}_{s}\right)$.
\end{definition}

Let $\kappa > 0$ denote the largest singular value of {$\mathbf{H}_{i_k}^H \mathbf{H}_{i_k}$} across all users, i.e.,
{
\begin{equation}
    \label{eq:sigma_HH}
    \kappa \triangleq \max_{i_k \in \mathcal{I}} \sigma_{\max}\left(\mathbf{H}_{i_k}^H \mathbf{H}_{i_k}\right).
\end{equation}
Furthermore, let $P_{\text{max}} \triangleq \max_{k} P_k$ represent the maximum transmit power among all BSs.}
We now present Lemma~\ref{lemma:lambda_ek}, which establishes a fundamental property of the matrix {$\mathbf{E}_{i_k}$}.
\begin{lemma}
    \label{lemma:lambda_ek}
    For any {$\mathbf{U}_{i_k}$} and {$\mathbf{V}_{i_k}\in \mathcal{P}_{k},i_k\in\mathcal{I}$}, the minimum eigenvalue of {$\mathbf{E}_{i_k}$} satisfies
    \[
    \lambda_{\min}({\mathbf{E}_{i_k}}) \geq \frac{{\sigma_{i_k}^{2}}}{P_{\max} \kappa + {\sigma_{i_k}^{2}}},
    \]
    where ${\mathbf{E}_{i_k}}$ is defined in~\eqref{eq:Ek}.
\end{lemma}
The proof is provided in Appendix~\ref{append:proof_lambda_ek}.
Before establishing the convergence of the A-MMMSE iterates to a stationary point, we first demonstrate in Lemma~\ref{lemma:limit_point} that at least one limit point exists.
The strong coupling between the update of ${\mathbf{W}_{i_k}}$ and the structure of ${\mathbf{E}_{i_k}}$, as established in Lemma~\ref{lemma:lambda_ek}, further implies an upper bound on the largest eigenvalue of {$\mathbf{W}_{i_k}$}. 
This eigenvalue bound plays a crucial role in establishing the existence of limit points, as detailed in the proof of Lemma~\ref{lemma:limit_point} provided in Appendix~\ref{append:proof_limit_point}.

\begin{lemma}
    \label{lemma:limit_point}
    The sequence of iterates generated by A-MMMSE has at least one limit point.
\end{lemma}

To proceed with the main convergence analysis, we require the following result on the smoothness of the objective function $f(\mathbf{U}, \mathbf{W}, \mathbf{V})$ with respect to $\mathbf{V}$. Let {$\bar{\alpha} \triangleq \max_{i_k\in\mathcal{I}} \alpha_{i_k}$} denote the maximum user priority weight and {$\sigma^2 \triangleq \min_{i_k\in \mathcal{I}} \sigma^2_{i_k}$}.
\begin{lemma}
\label{lemma:L_V}
 The function $f(\mathbf{U}, \mathbf{W}, \mathbf{V})$ is $L$-smooth in $\mathbf{V}$ with $L \leq L_v \triangleq \frac{2\bar{\alpha}{KI}\kappa}{\sigma^2}$, for any $\mathbf{U}$ and $\mathbf{W}$ produced by the A-MMMSE algorithm.
\end{lemma}
The detailed proof is relegated to Appendix~\ref{append:proof_lemma_L_V}.
Lemma~\ref{lemma:L_V} establishes an upper bound $L_v$ on the Lipschitz constant of $f$ with respect to $ \mathbf{V} $. 
Since we replace the update of $\mathbf{V}$ from BCM to BCGD, while updates of $\mathbf{U}$ and $\mathbf{W}$ consistently minimize their respective subproblems, the sequence $\left\{(\mathbf{U}^t, \mathbf{W}^t, \mathbf{V}^t)\right\}_t$ is guaranteed to exhibit a monotonically non-increasing objective function $f$, provided that $\gamma_v \leq \frac{1}{L_v}$.
Based on Lemma~\ref{lemma:lambda_ek}~--~\ref{lemma:L_V}, we can establish that the A-MMMSE algorithm retains the same convergence property as WMMSE~\cite{shi2011iteratively}.
\begin{theorem}
    \label{theo:convergence}
    Any limit point $(\overline{\mathbf{U}}, \overline{\mathbf{W}}, \overline{\mathbf{V}})$ of the iterates generated by the A-MMMSE algorithm, with a single PGD step of size $0 < \gamma \leq \frac{\sigma^2}{2 \bar{\alpha} {KI} \kappa}$, is a stationary point of problem~\eqref{prob:wmmse_prob} and the corresponding $\overline{\mathbf{V}}$ is a stationary point of the original problem~\eqref{prob:ori_prob}. 
    Conversely, if $\overline{\mathbf{V}}$ is a stationary point of~\eqref{prob:ori_prob}, then $(\overline{\mathbf{U}}, \overline{\mathbf{W}}, \overline{\mathbf{V}})$ is a stationary point of~\eqref{prob:wmmse_prob}.
\end{theorem}

\begin{proof}
Let $\mathbf{Z}_{u}^{t} \triangleq (\mathbf{U}^{t+1}, \mathbf{W}^{t}, \mathbf{V}^{t})$, $\mathbf{Z}_{w}^{t} \triangleq (\mathbf{U}^{t+1}, \mathbf{W}^{t+1}, \mathbf{V}^{t})$, and $\mathbf{Z}_{v}^{t} \triangleq (\mathbf{U}^{t+1}, \mathbf{W}^{t+1}, \mathbf{V}^{t+1})$, where $t$ denotes the iteration index. 
At each A-MMMSE iteration, the objective function $f$ is non-increasing. Specifically,
\begin{equation}
\label{prob:min_zu}
f\left(\mathbf{Z}_{u}^{t}\right) \triangleq \min _{\mathbf{\xi}} f\left(\mathbf{\xi}, \mathbf{W}^{t}, \mathbf{V}^{t}\right)
\end{equation}
and
\begin{equation}
\label{prob:min_zw}
f\left(\mathbf{Z}_{w}^{t}\right) \triangleq \min _{\xi} f\left(\mathbf{U}^{t+1}, \mathbf{\xi}, \mathbf{V}^{t}\right),
\end{equation}
which imply
\begin{equation}
\label{ineq:mono_1}
f\left(\mathbf{Z}_{v}^{t-1}\right) \geq f\left(\mathbf{Z}_{u}^{t}\right) \geq f\left(\mathbf{Z}_{w}^{t}\right) .
\end{equation}
Applying a single PGD step with step size \(\gamma \leq \frac{\sigma^2}{2\bar{\alpha} {KI} \kappa}\) to the subproblem
\begin{equation}
\label{prob:min_zp}
\begin{array}{l}
\underset{\mathbf{\xi}}{\min}  f\left(\mathbf{U}^{t+1}, \mathbf{W}^{t+1}, \mathbf{\xi}\right) \\
\text { s.t. } \operatorname{Tr}\left(\mathbf{\xi} \mathbf{\xi}^{H}\right) \leq {P_{\xi}},
\end{array}
\end{equation}
{where \(P_{\xi}\) denotes the maximum power budget that \(\boldsymbol{\xi}\) must satisfy}, and leveraging the $L$-smoothness of $f(\mathbf{U}^{t+1}, \mathbf{W}^{t+1}, \mathbf{V})$ with $L \leq \frac{2\bar{\alpha} {KI} \kappa}{\sigma^2}$ (as established in Lemma~\ref{lemma:L_V}), it follows from~\cite{nesterov2013introductory} that
\begin{equation}
\label{ineq:mono_2}
f\left(\mathbf{Z}_{w}^{t}\right) \geq f\left(\mathbf{Z}_{v}^{t}\right).
\end{equation}
Combining~\eqref{ineq:mono_1} and~\eqref{ineq:mono_2} yields
\begin{equation*}
f\left(\mathbf{Z}_{v}^{t-1}\right) \geq f\left(\mathbf{Z}_{u}^{t}\right) \geq f\left(\mathbf{Z}_{w}^{t}\right) \geq f\left(\mathbf{Z}_{v}^{t}\right) \quad  \forall t.
\end{equation*}

Let $\overline{\mathbf{X}} = (\overline{\mathbf{U}}, \overline{\mathbf{W}}, \overline{\mathbf{V}})$ be a limit point of the sequence $\{\mathbf{X}^t\}_t$, as guaranteed by Lemma~\ref{lemma:limit_point}. 
This implies that the sequences $\{\mathbf{U}^t\}_t$, $\{\mathbf{W}^t\}_t$, and $\{\mathbf{V}^t\}_t$ admit limit points $\overline{\mathbf{U}}$, $\overline{\mathbf{W}}$, and $\overline{\mathbf{V}}$, respectively.
For $\mathbf{U}$ and $\mathbf{W}$, the convergence proof follows directly from the analysis of BCD methods in~\cite{bertsekas1997nonlinear}, since subproblems~\eqref{prob:min_zu} and~\eqref{prob:min_zw} are solved optimally. Restricting attention to a convergent subsequence $\{t_j\}$, we have
\begin{equation*}
f\left(\mathbf{Z}_{u}^{t_{j}}\right) \leq f\left(\mathbf{U}, \mathbf{W}^{t_{j}}, \mathbf{V}^{t_{j}}\right), \quad \forall \mathbf{U} \text { and } \forall j \geq 1
\end{equation*}
\begin{equation*}
f\left(\mathbf{Z}_{w}^{t_{j}}\right) \leq f\left(\mathbf{U}^{t_{j}+1}, \mathbf{W}, \mathbf{V}^{t_{j}}\right),\quad \forall \mathbf{W} \text { and } \forall j \geq 1 .
\end{equation*}
Taking the limit as $j \to \infty$, we obtain
\begin{equation*}
f(\overline{\mathbf{X}}) \leq f(\mathbf{U}, \overline{\mathbf{W}}, \overline{\mathbf{V}}),\quad \forall \mathbf{U},
\end{equation*}
\begin{equation*}
f(\overline{\mathbf{X}}) \leq f(\overline{\mathbf{U}}, \mathbf{W}, \overline{\mathbf{V}}),\quad \forall \mathbf{W}.
\end{equation*}
Since \(f(\mathbf{U}, \overline{\mathbf{W}}, \overline{\mathbf{V}})\) and \(f(\overline{\mathbf{U}}, \mathbf{W}, \overline{\mathbf{V}})\) are differentiable in \(\mathbf{U}\) and \(\mathbf{W}\), respectively, and the optima lie in the interior of the domain, the gradients must vanish at \(\overline{\mathbf{X}}\):
\begin{equation}
\label{ineq:u_stat}
\nabla_{u} f(\overline{\mathbf{X}})=\mathbf{0} \text { and }\left\langle\nabla_{u} f(\overline{\mathbf{X}}), \mathbf{U}-\overline{\mathbf{U}}\right\rangle=0,
\end{equation}
\begin{equation}
\label{ineq:w_stat}
\nabla_{w} f(\overline{\mathbf{X}})=\mathbf{0} \text { and }\left\langle\nabla_{w} f(\overline{\mathbf{X}}), \mathbf{W}-\overline{\mathbf{W}}\right\rangle=0.
\end{equation}
Hence, $\overline{\mathbf{U}}$ and $\overline{\mathbf{W}}$ are stationary points of~\eqref{prob:min_zu} and~\eqref{prob:min_zw}, respectively.
It remains to show that \(\overline{\mathbf{V}}\) is a stationary point of~\eqref{prob:min_zp}. 
Let \(\mathbf{V}^+\) denote the next gradient descent iterate starting from \(\overline{\mathbf{X}}\).
By Lemma~\ref{lemma:L_V} and the convexity of \(f\) in \(\mathbf{V}\), it follows from~\cite{nesterov2013introductory} that for any \(0 < \gamma_v \leq 1/L_v\),
\begin{equation*}
f\left(\overline{\mathbf{U}}, \overline{\mathbf{W}}, \mathbf{V}^{+}\right) \leq f(\overline{\mathbf{U}}, \overline{\mathbf{W}}, \overline{\mathbf{V}})-\frac{\gamma_{v}\left\|\nabla_{v} f(\overline{\mathbf{X}})\right\|^{2}}{2}.
\end{equation*}
Therefore, it must hold that
\begin{equation*}
\left\langle\nabla_{v} f(\overline{\mathbf{X}}), \mathbf{V}-\overline{\mathbf{V}}\right\rangle\geq 0,
\end{equation*}
since otherwise we would have
\begin{equation}
\label{ineq:p_stat}
f\left(\overline{\mathbf{U}}, \overline{\mathbf{W}}, \mathbf{V}^{+}\right)<f(\overline{\mathbf{U}}, \overline{\mathbf{W}}, \overline{\mathbf{V}}),
\end{equation}
contradicting the assumption that $\overline{\mathbf{V}}$ is a limit point of $\{\mathbf{V}^t\}$.
Combining~\eqref{ineq:u_stat},~\eqref{ineq:w_stat}, and~\eqref{ineq:p_stat}, we conclude that $\overline{\mathbf{X}}$ is a stationary point of problem~\eqref{prob:wmmse_prob}, i.e.,
\begin{equation}
\left\langle\nabla f(\overline{\mathbf{X}}), \mathbf{X}-\overline{\mathbf{X}}\right\rangle\geq 0.
\end{equation}

Finally, we note that the equivalence between problems~\eqref{prob:ori_prob} and~\eqref{prob:wmmse_prob}, established in~\cite{shi2011iteratively}, remains unaffected by the algorithmic approach used to solve~\eqref{prob:wmmse_prob}. 
Therefore, to show that $\overline{\mathbf{V}}$ is a stationary point of~\eqref{prob:ori_prob} if and only if $\overline{\mathbf{X}} = (\overline{\mathbf{U}}, \overline{\mathbf{W}}, \overline{\mathbf{V}})$ is a stationary point of~\eqref{prob:wmmse_prob} for some $\overline{\mathbf{U}}$ and $\overline{\mathbf{W}}$, the second part of the proof in~\cite{shi2011iteratively} applies verbatim.
\end{proof}

\section{Simulation Results}
\label{sec:experiments}

\subsection{Setup}
    {We consider a downlink MU-MIMO system consisting of $K$ BSs. 
    Each BS is equipped with $M$ antennas and serves $I$ users in its own cell, where each user is equipped with $N$ antennas.
    The simulation setup is configured as follows:}
\begin{itemize}
    \item All weights {$\alpha_{i_k}$ for $i_k\in\mathcal{I}$} are set to 1.  
    \item The bisection method in \texttt{WMMSE}, \texttt{MMMSE}, {\texttt{R-WMMSE}}, and {\texttt{R-MMMSE}} terminates when the width of the search interval falls below $10^{-4}$ or when the number of iterations exceeds 100.
    \item The precoding matrices $\mathbf{V}$ for all algorithms are initialized to feasible values via projection of matrices with entries drawn from the complex Gaussian distribution with zero mean and unit variance.
    \item The channel coefficients are modeled as i.i.d. circularly symmetric complex Gaussian (Rayleigh) random variables with zero mean and unit variance.
    \item The sum power budget of each BS is set to {$P_{k} = 10 \; [\text{W}]$, for $k = 1, \dots, K$}.
    \item The noise power is assumed equal for all users and is defined as {$\sigma^{2}=10^{\frac{1}{IK} \sum_{(l,j)} \log _{10} \frac{1}{N}\left\|\mathbf{H}_{l_jj}\right\|_{F}^{2}} \times 10^{-\frac{\mathrm{SNR}}{10}}$}, where signal-to-noise ratio (SNR) denotes the average received signal-to-noise ratio per user without precoding.
    \item All \texttt{MMMSE}-based methods use $\epsilon_1=0.1$.
    \item All simulation results are averaged over 100 randomly generated channel realizations.
    % \item {
    % The extrapolation parameter $\omega$ in~\eqref{eq:extro_v} is computed as
    % \begin{equation}
    % \omega^t = \max\left\{\frac{t-2}{t+1}, 0\right\}, \quad \forall t \geq 1,
    % \end{equation}
    % following the setting in \texttt{NQT}~\cite{shen2024accelerating}.
    % }
    % \item {
    % The step size $\gamma$ for $V_{i_k}$ in~\eqref{eq:bcgd_Vk} is computed as
    % \begin{equation}
    % \gamma^t_{i_k} = 1/\lambda_{\max}\left(\sum_{(l,j)} 2\alpha_{l_j}\mathbf{H}_{l_jk}^H\mathbf{U}_{l_j}^t\mathbf{W}_{l_j}^t\mathbf{U}_{l_j}^{t}{}^{H}\mathbf{H}_{l_jk}\right).
    % \end{equation}}
\end{itemize}

All experiments and algorithms in this work were implemented in Python 3.10.0 and executed on a computing platform equipped with an Intel Xeon 2.10 GHz CPU and an NVIDIA RTX A6000 GPU.
Our code is available at \href{https://github.com/NetSysOpt/A-MMMSE}{https://github.com/NetSysOpt/A-MMMSE}.

\subsection{Baselines}
In our simulations, the following algorithms are compared:  
\begin{itemize}
    \item \texttt{WMMSE}~\cite{shi2011iteratively}:A classical algorithm for WSR maximization, as detailed in Section~\ref{sec:wmmse}.
    \item \texttt{NQT}~\cite{shen2024accelerating}: An improved method based on quadratic transform that eliminates matrix inversion during precoder updates by optimizing a nonhomogeneous bound. 
    This approach, however, is only applicable when the number of data streams satisfies $d=1$.
    \item \texttt{R-WMMSE}~\cite{zhao2023rethinking}: An accelerated variant of \texttt{WMMSE} that reduces the computational complexity from cubic to linear in the number of BS antennas by leveraging low-dimensional subspace structures of stationary points.
    \item \texttt{MMMSE}: A two-stage method that first solves problem~\eqref{prob:mmse_prob}, then switches to problem~\eqref{prob:wmmse_prob} once the relative change in WSR falls below $\epsilon_1$, and terminates when the change drops below $\epsilon_2$.
    \item \texttt{R-MMMSE}: An extension of \texttt{MMMSE} that integrates the low-dimensional subspace strategy used in \texttt{R-WMMSE}.
    \item {\texttt{A-WMMSE}: An enhanced version of \texttt{WMMSE} that employs the update rule in~\eqref{eq:bcd_proxi_linear} for $\mathbf{V}$ and incorporates an extrapolation technique.} 
    \item \texttt{A-MMMSE}: An enhanced version of \texttt{MMMSE} that employs the update rule in~\eqref{eq:bcd_proxi_linear} for $\mathbf{V}$ and incorporates an extrapolation technique. 
\end{itemize}

\subsection{Convergence Performance}
We first evaluate the convergence behavior of each algorithm by fixing the number of iterations and plotting the WSR, measured in bits per channel use (bpcu), for all methods in Fig.~\ref{fig:convg}. 
The vertical dashed lines indicate the average iteration number where the corresponding \texttt{MMMSE}-based algorithms transition from Stage I to Stage II. 
Specifically, iterations before the vertical line correspond to solving problem~\eqref{prob:mmse_prob}, while iterations after the line solve problem~\eqref{prob:wmmse_prob}.
In Fig.~\ref{fig:convg_sub1} and Fig.~\ref{fig:convg_sub2}, the transition iteration numbers for \texttt{MMMSE} and \texttt{R-MMMSE} coincide, resulting in overlapping vertical lines. 
Due to the Lagrangian dual transformation employed in \texttt{NQT}, this method is only applicable to single-data-stream scenarios ($d = 1$). 

\begin{figure}
    \centering
    \subfigure[{$K=2$, $I=8$, $M=128$, $N=4$, $d=1$, SNR=0dB.}]{\includegraphics[width=0.48\textwidth]{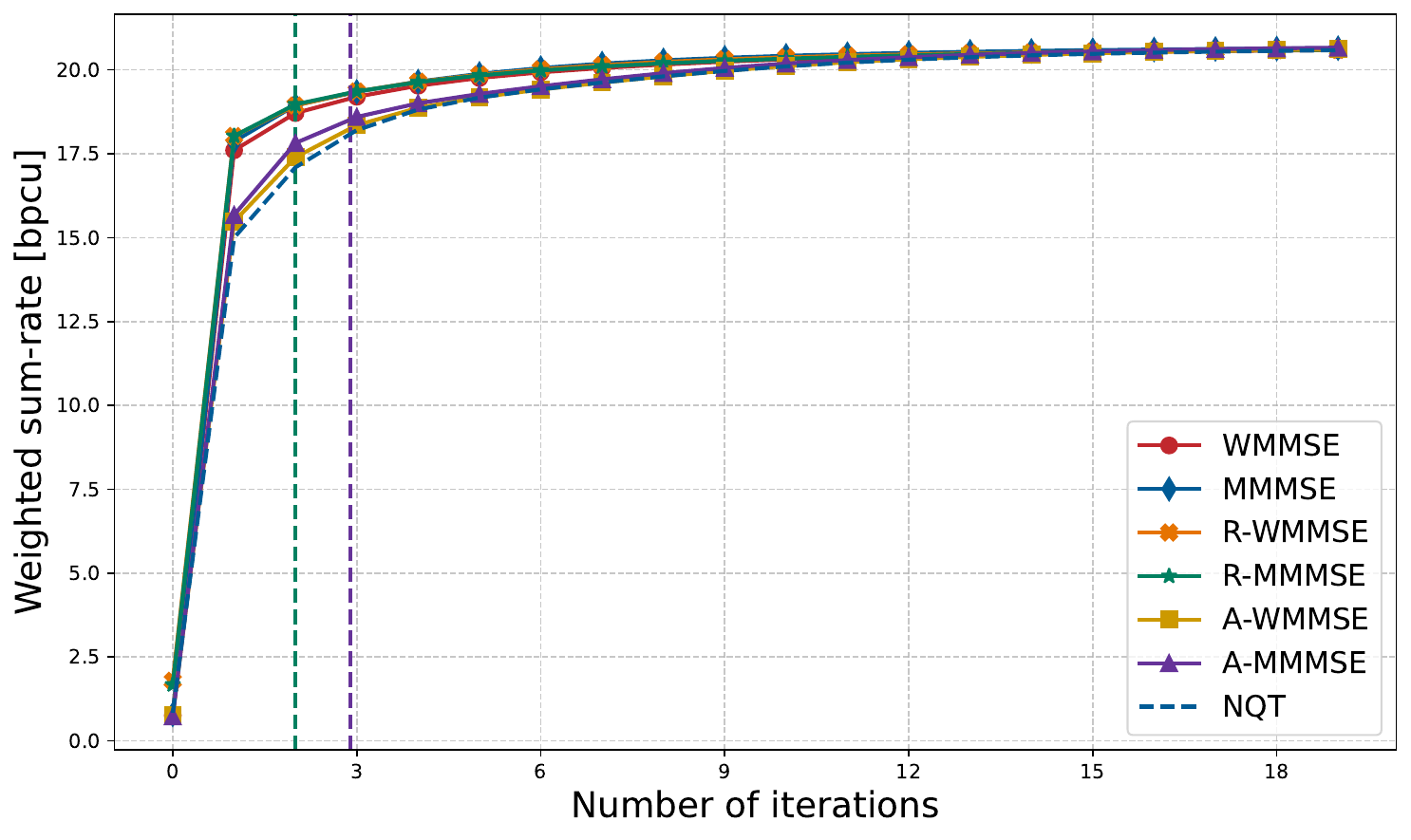} \label{fig:convg_sub1}}
    \subfigure[{$K=4$, $I=12$, $M=256$, $N=4$, $d=4$, SNR=10dB.}]{\includegraphics[width=0.48\textwidth]{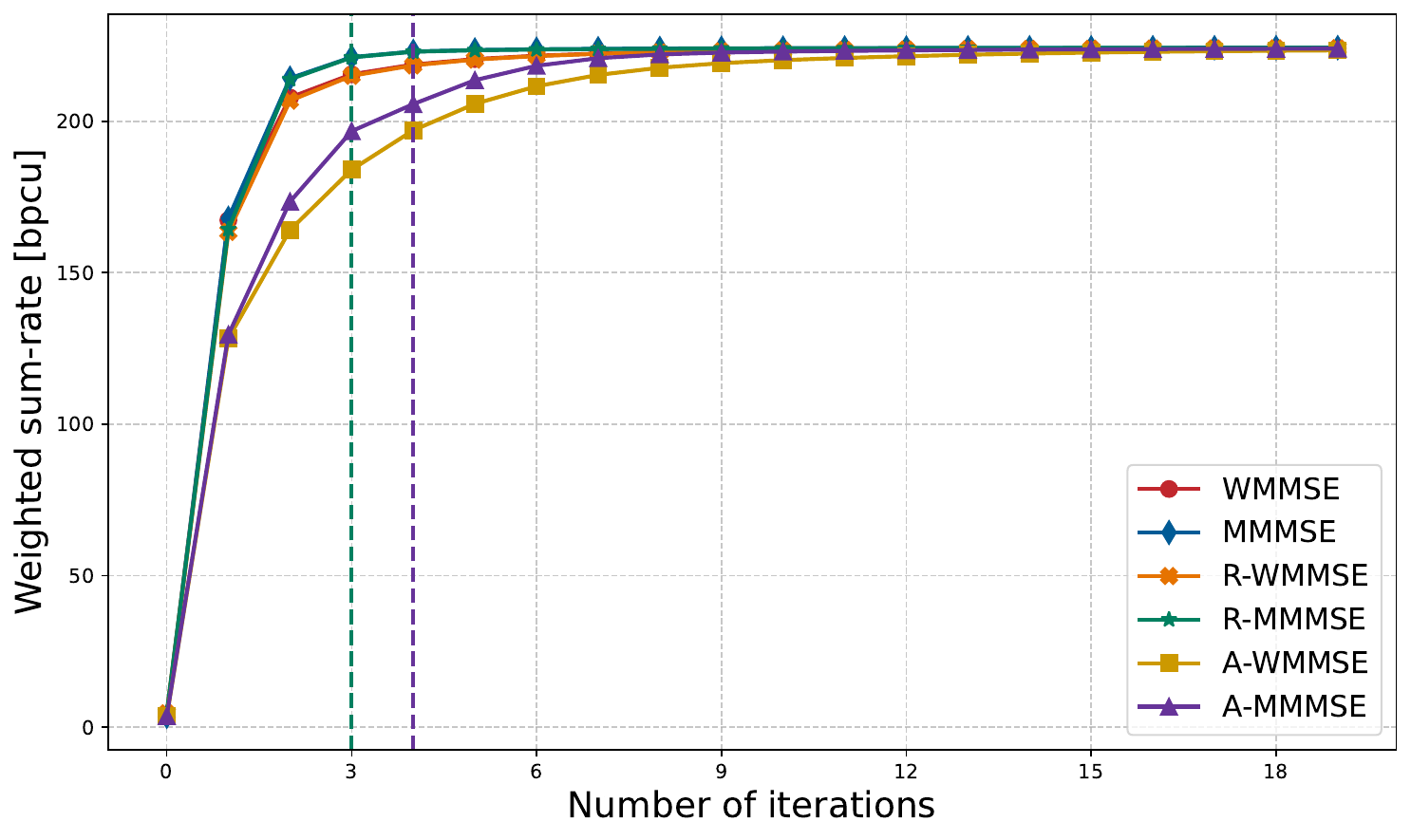} \label{fig:convg_sub2}}
    \caption{{Convergence performance of different algorithms with the number of iterations fixed at 20.}} 
\label{fig:convg}
\end{figure}

{
The results show that all algorithms converge to almost the same optimal WSR value.
Compared with the traditional BCM method, all gradient-based algorithms (\texttt{NQT}, \texttt{A-WMMSE}, \texttt{A-MMMSE}) exhibit relatively slower convergence rates due to their first-order update strategy for $\mathbf{V}$.
Our proposed mixed weighted-unweighted warm-start method consistently improves convergence, particularly in larger-scale scenarios. 
Specifically, for the case $d = 1$ (Fig.~\ref{fig:convg_sub1}), we have $M = K \cdot I \cdot N$ in this configuration, so \texttt{R-WMMSE} does not demonstrate a faster convergence rate than \texttt{WMMSE}. 
Because the primary modification of the mixed weighted-unweighted warm-start method lies in $\mathbf{W}$, the observed improvement is relatively modest when $d = 1$. 
Nevertheless, \texttt{A-MMMSE} still exhibits a faster convergence rate than \texttt{A-WMMSE}. 
Meanwhile, due to the equivalence between QT and WMMSE, the convergence rates of \texttt{NQT} and \texttt{A-WMMSE} are similar and remain slower than that of \texttt{A-MMMSE}.
For the case $d = 4$ (Fig.~\ref{fig:convg_sub2}), both the number of BS antennas $M$ and the number of users $I$ are increased. 
In this setting, the total number of receive antennas across users remains close to the number of transmit antennas per BS, so \texttt{R-WMMSE} again exhibits a convergence rate similar to that of \texttt{WMMSE}. 
Under these conditions, the proposed mixed weighted-unweighted warm-start method yields a substantial improvement in convergence speed. 
Both \texttt{MMMSE} and \texttt{R-MMMSE} achieve similarly fast and optimal convergence rates. 
The vertical dashed lines indicate that only a few iterations of the first-stage MMSE are needed to achieve a significant improvement in convergence rate.
Furthermore, compared with \texttt{A-WMMSE}, which relies solely on hybrid BCD acceleration, \texttt{A-MMMSE} demonstrates a more pronounced convergence rate improvement than in the low-stream-number case.
}

\begin{figure}
    \centering
    \subfigure[{$K=2$, $I=8$, $M=128$, $N=4$, $d=1$, SNR=0dB.}]{\includegraphics[width=0.48\textwidth]{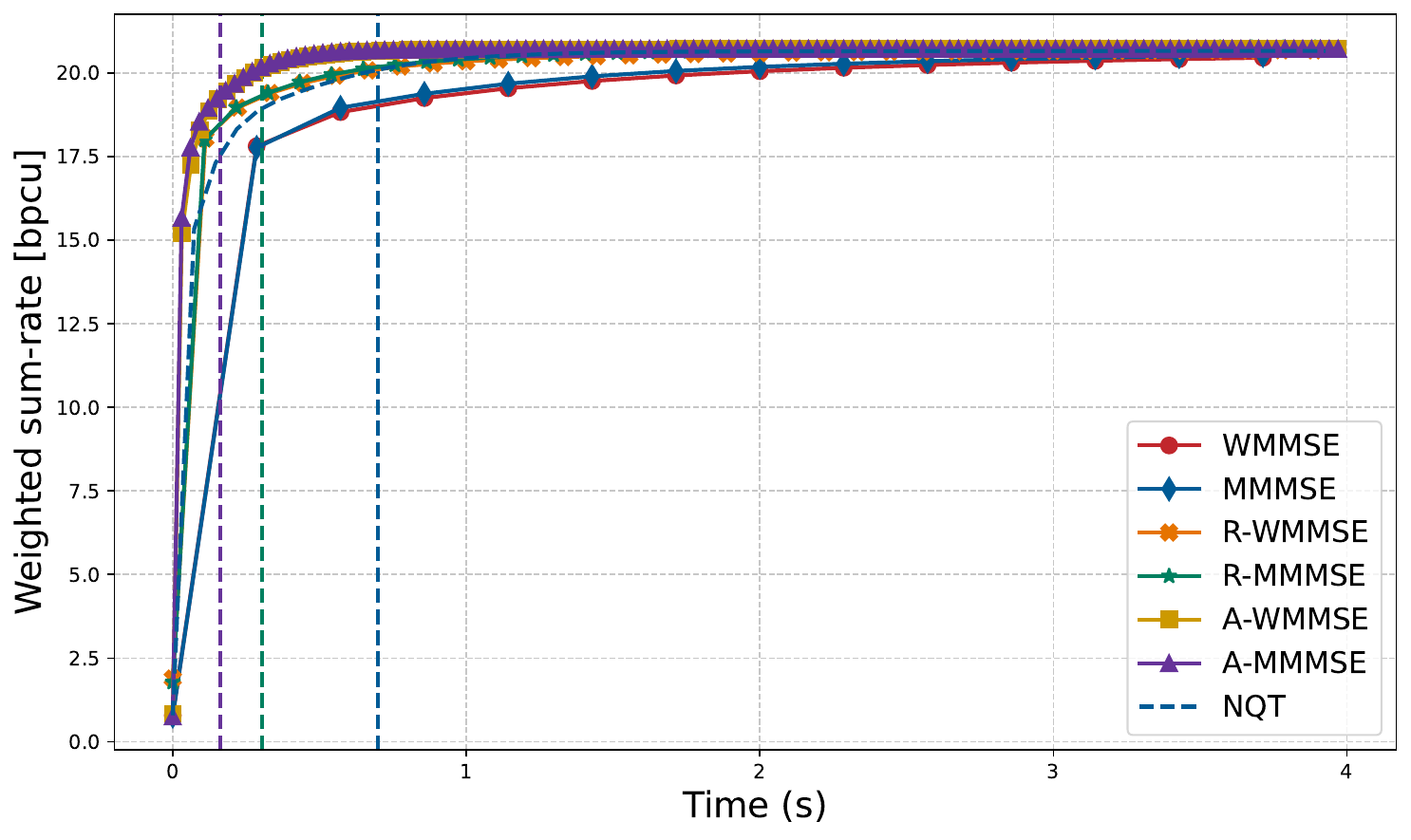} \label{fig:time_convg_sub1}}
    \subfigure[{$K=4$, $I=12$, $M=256$, $N=4$, $d=4$, SNR=10dB.}]{\includegraphics[width=0.48\textwidth]{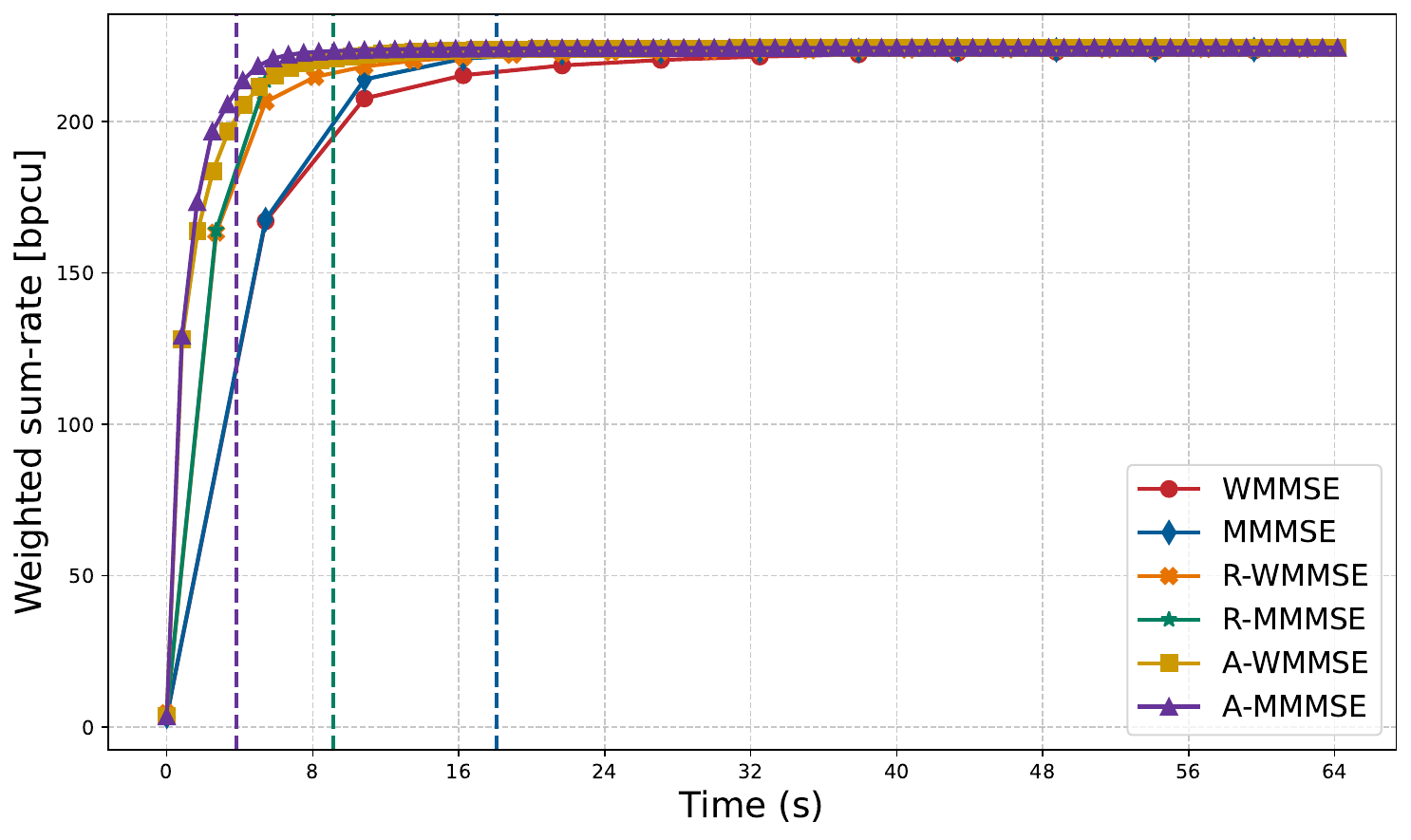} \label{fig:time_convg_sub2}}
    \caption{{Convergence performance of different algorithms under a fixed runtime.}} 
\label{fig:time_convg}
\end{figure}

{
Next, we present the convergence performance of all algorithms under a fixed runtime in Fig.~\ref{fig:time_convg}.
Although gradient-based methods exhibit slower convergence rates when the number of iterations is fixed, they demonstrate significantly faster convergence under a fixed runtime due to their lower per-iteration computational cost. 
For the case $d = 1$, \texttt{A-MMMSE}, \texttt{A-WMMSE}, and \texttt{NQT} have comparable per-iteration costs and thus exhibit similar convergence behavior, all of which significantly outperform the traditional BCM-based approach. 
% The proposed mixed weighted-unweighted warm-start method yields a modest improvement when $d = 1$, though its advantage remains limited. 
As the system scale increases, our proposed \texttt{A-MMMSE} achieves the fastest convergence. 
Consistent with the observations in Fig.~\ref{fig:convg_sub2}, the mixed weighted-unweighted warm-start method brings substantial improvements to all WMMSE-based methods. 
The positions of the vertical dashed lines indicate that algorithms transitioning earlier to the second stage have lower per-iteration runtime. 
Consequently, \texttt{A-MMMSE} incurs lower cost than \texttt{R-MMMSE}, while \texttt{MMMSE} has the highest cost. 
Furthermore, although \texttt{R-WMMSE} exhibits a convergence rate similar to that of \texttt{WMMSE} in Fig.~\ref{fig:convg_sub2}, its computational advantage in scenarios with a large number of BS antennas leads to significantly faster convergence in terms of runtime. 
\texttt{R-MMMSE} achieves a faster convergence rate than \texttt{R-WMMSE}, but due to the cubic complexity of its subspace structure in the total number of users, it remains slower than \texttt{A-WMMSE} and \texttt{A-MMMSE}. 
The combination of the hybrid BCD method and the mixed weighted-unweighted warm-start method ultimately enables \texttt{A-MMMSE} to achieve the fastest convergence rate.
}

\begin{figure}[htbp]
    \centering
    \subfigure[{Average CPU time.}]{\includegraphics[width=0.48\textwidth]{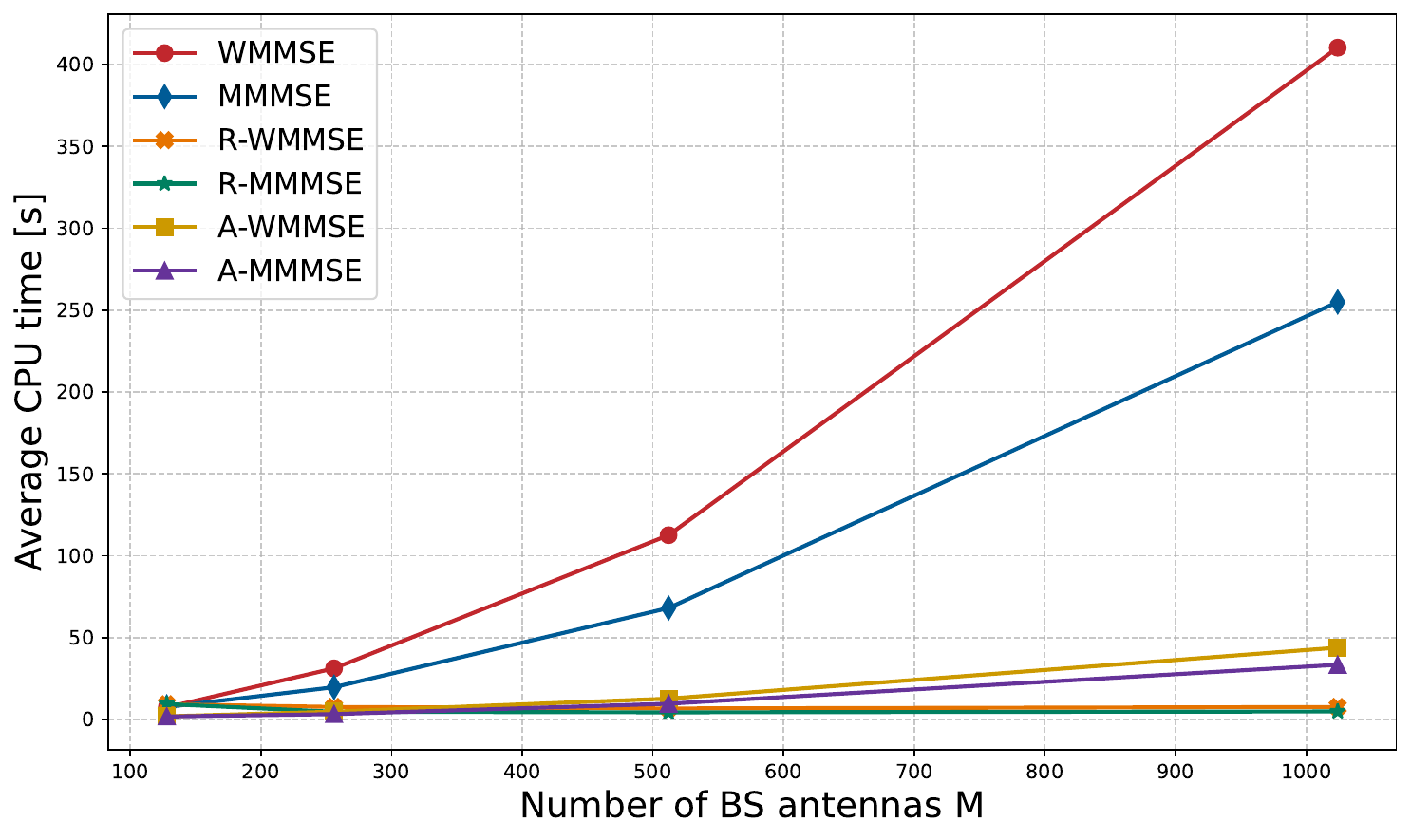} \label{fig:t_cpu_time}}
    \subfigure[{Weighted sum-rate.}]{\includegraphics[width=0.48\textwidth]{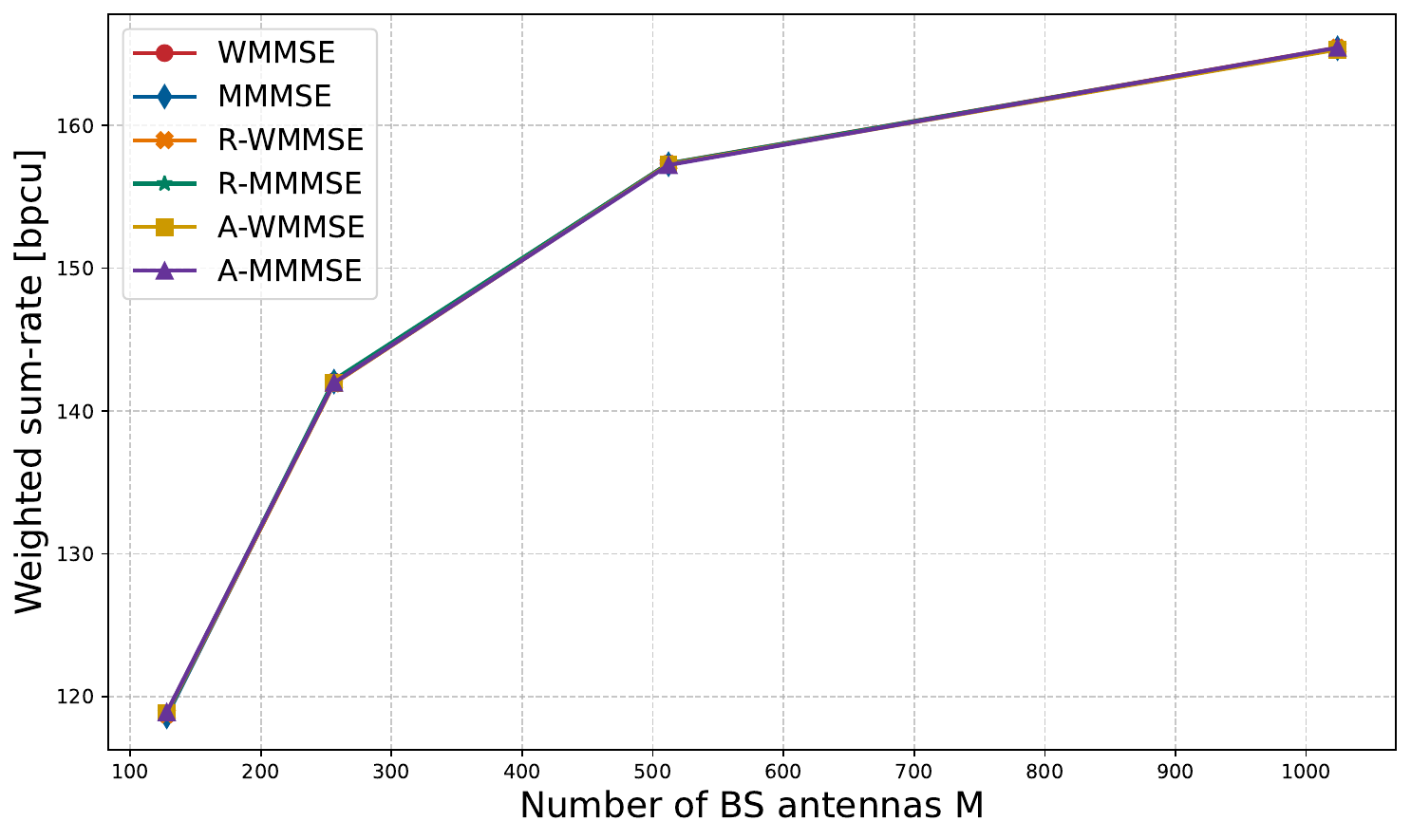} \label{fig:t_sum_rate}}
    \subfigure[{Iterations.}]{\includegraphics[width=0.48\textwidth]{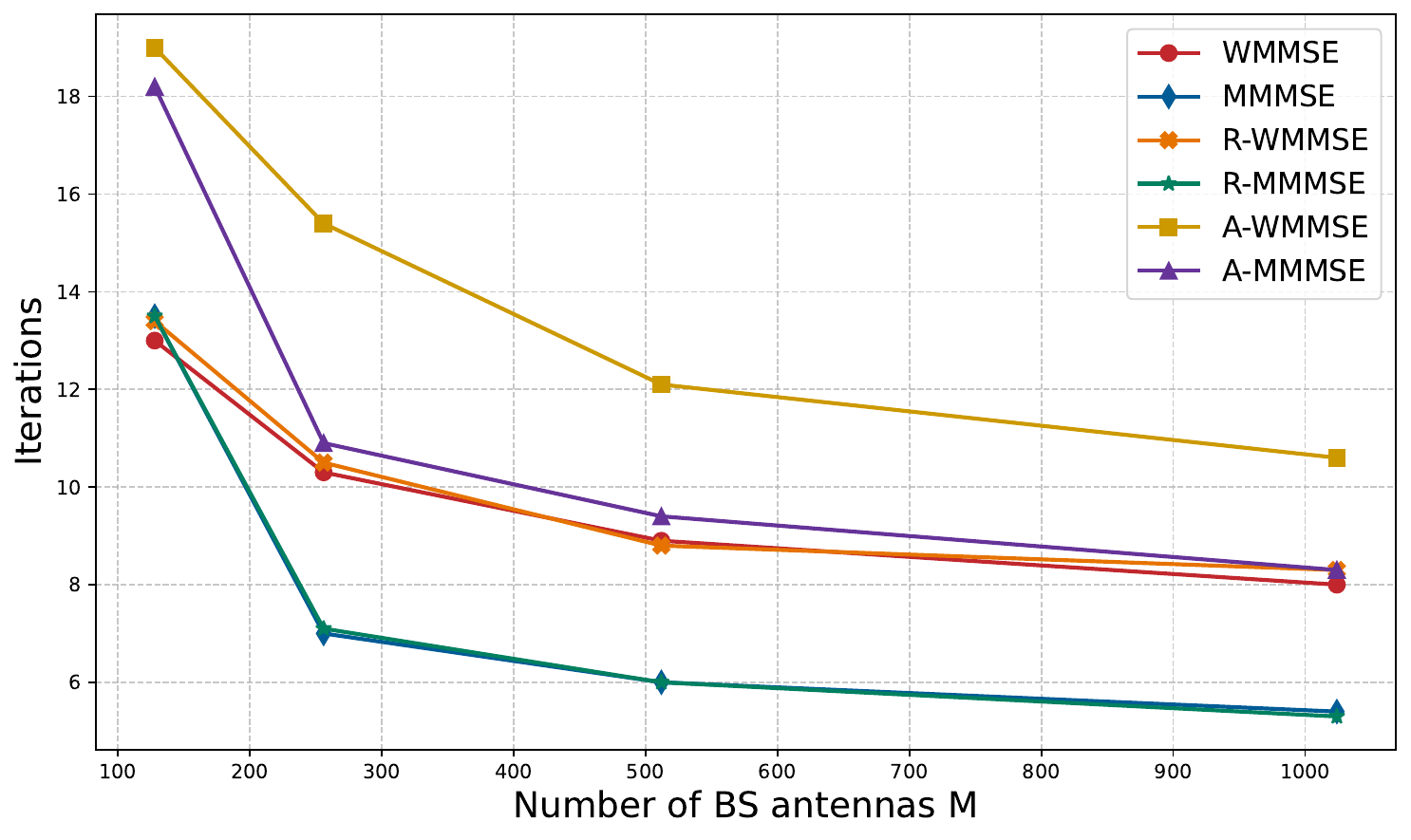} \label{fig:t_ite}}
    \subfigure[{CPU time reduction.}]{\includegraphics[width=0.48\textwidth]{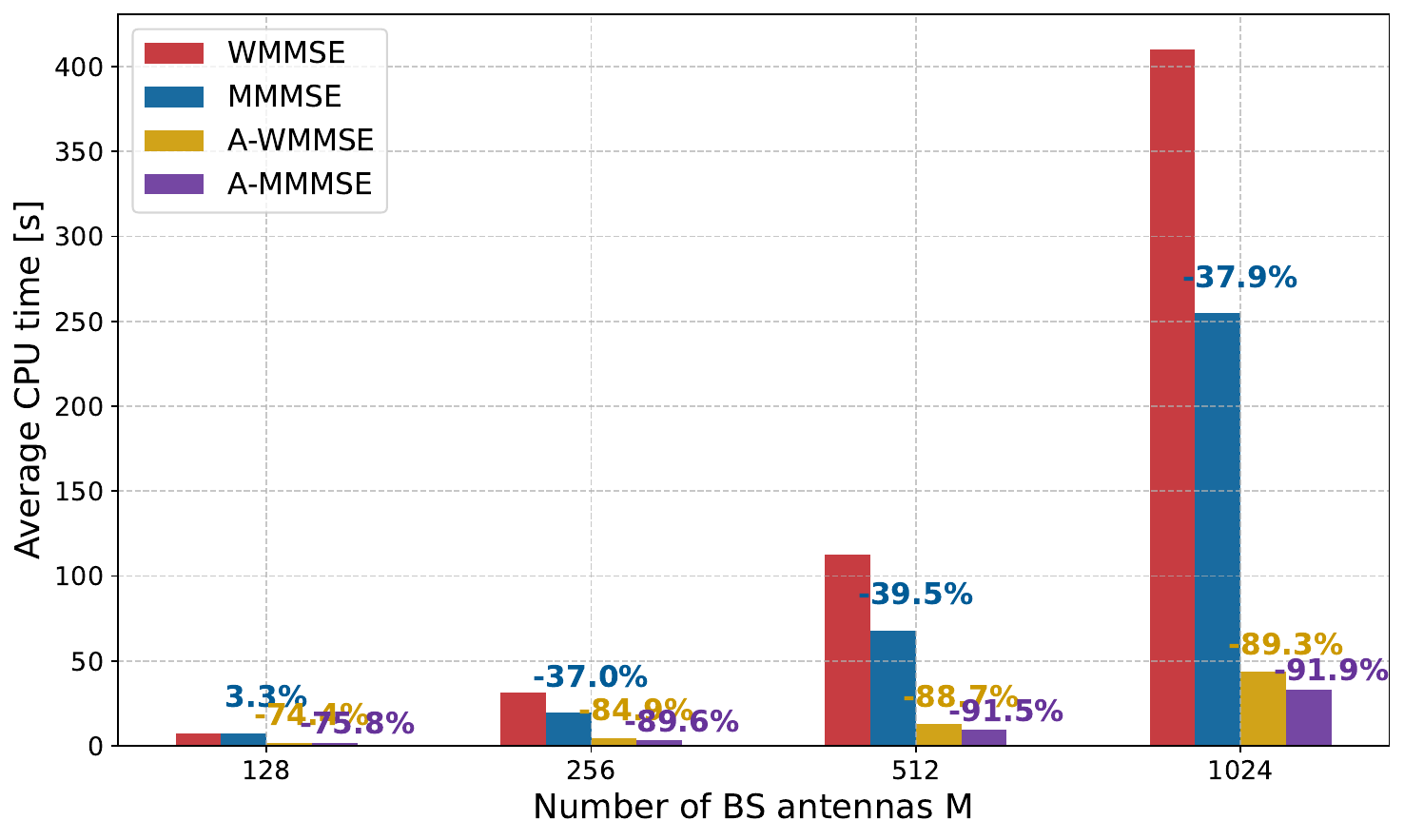} \label{fig:t_rduc}}
    \caption{{Performance evaluation of algorithms under varying numbers of BS antennas, with $K=2$, $I=16$, $N=4$, $d=4$ and $\text{SNR}=10$dB.}}
\label{fig:t_performance}
\end{figure}

\begin{figure}[htbp]
    \centering
    \subfigure[{Average CPU time.}]{\includegraphics[width=0.48\textwidth]{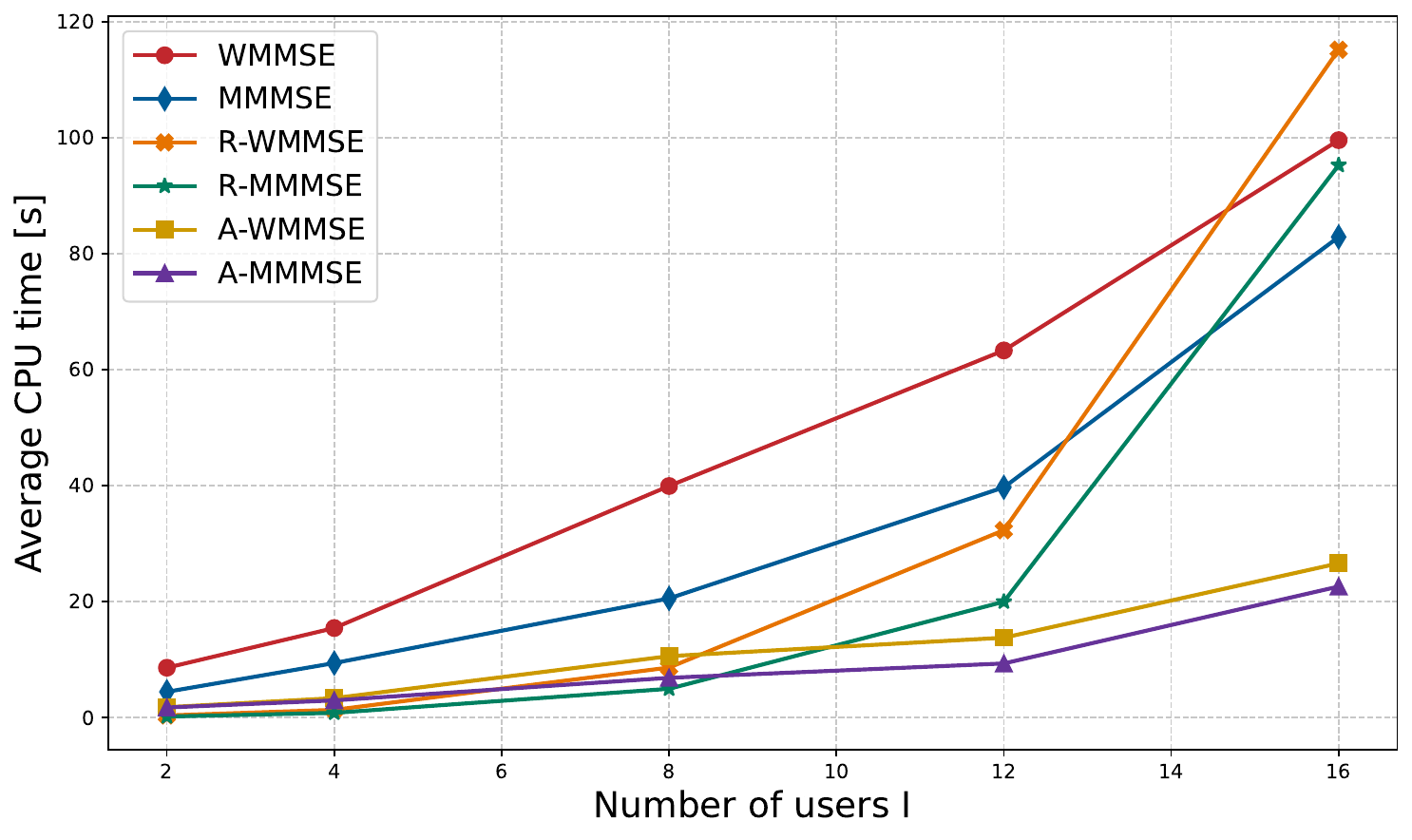} \label{fig:i_cpu_time}}
    \subfigure[{Weighted sum-rate.}]{\includegraphics[width=0.48\textwidth]{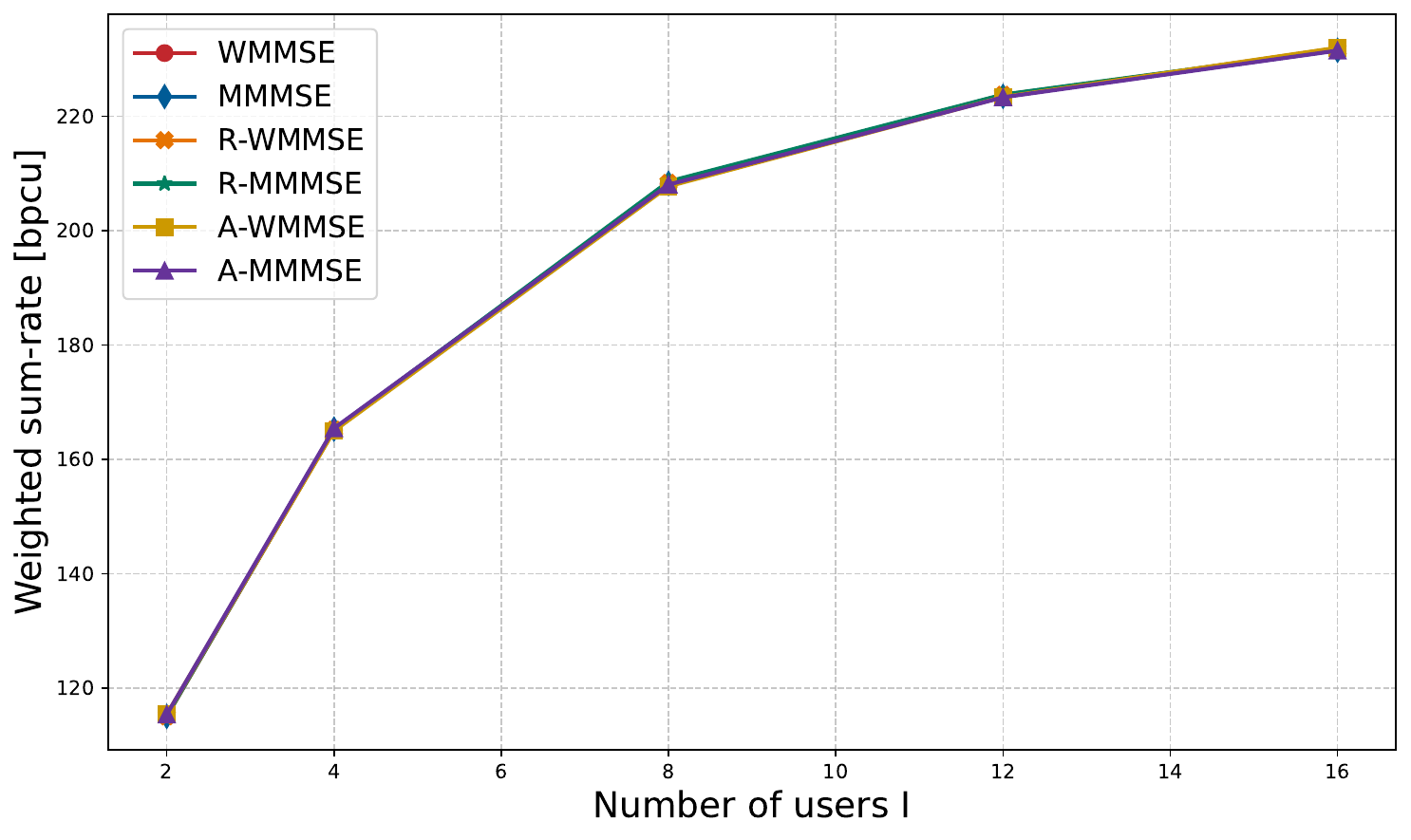} \label{fig:i_sum_rate}}
    \subfigure[{CPU time reduction.}]{\includegraphics[width=0.48\textwidth]{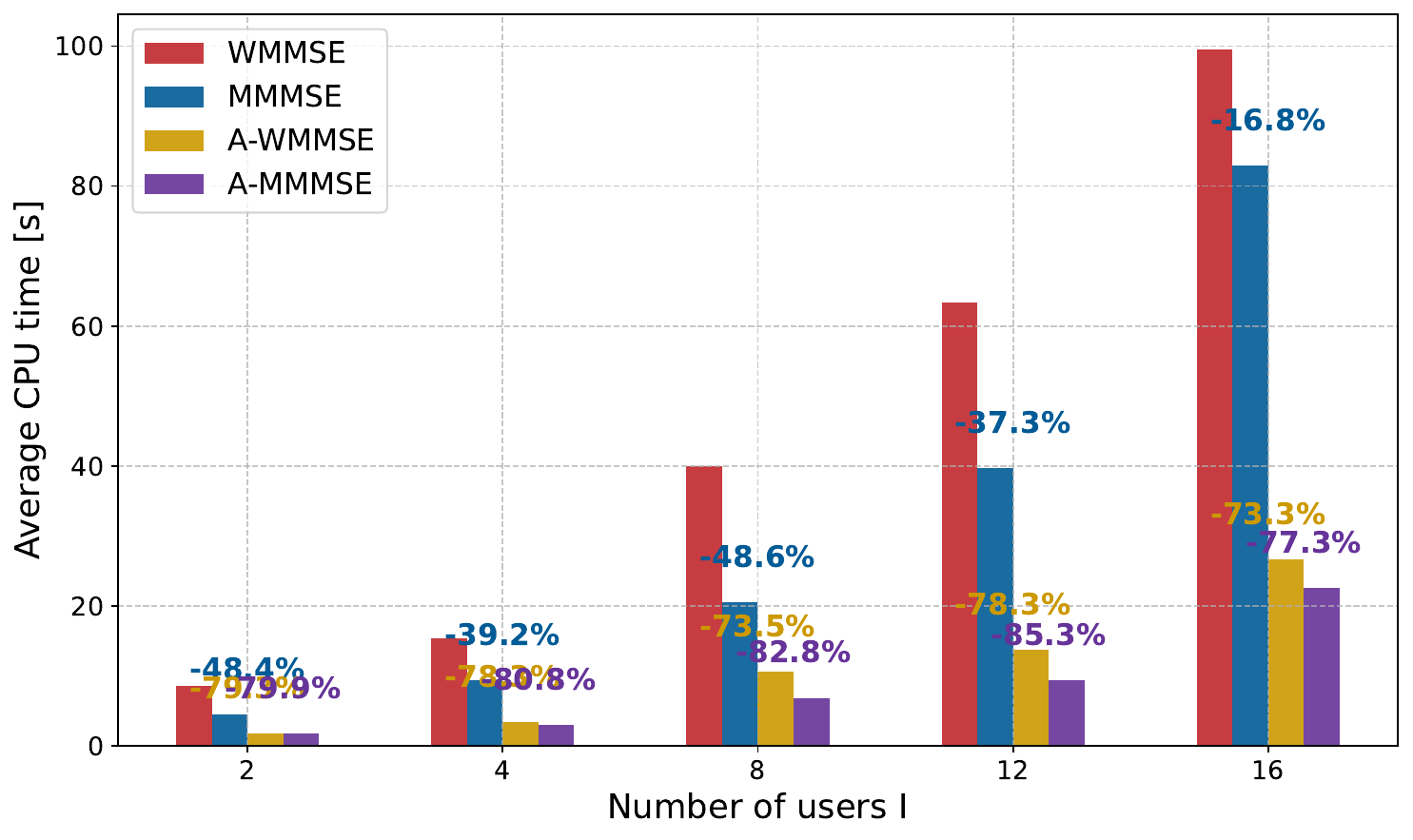} \label{fig:i_reduc}}
    \caption{{Performance evaluation of algorithms under varying numbers of users, with $K=4$, $M=256$, $N=4$, $\text{SNR}=10$ dB and $d=4$.}}
\label{fig:i_performance}
\end{figure}

\subsection{Execution Time on CPUs} \label{sec:cpu_time_sec}
We now adopt a consistent stopping criterion across all algorithms by setting $\epsilon_2=0.001$, and compare their CPU runtime, iteration count, and final achieved WSR under varying numbers of BS antennas, users, and SNR values.

\subsubsection{{Performance with Different $M$}}
First, we evaluate performance as the number of BS antennas $M$ varies, with the number of cells $K = 2$, the number of users fixed at $I = 16$, the number of data streams at $d = 4$, and the SNR at 10 dB.
As shown in Fig.~\ref{fig:t_sum_rate}, all algorithms achieve almost the same final WSR in this configuration.

{Fig.~\ref{fig:t_cpu_time} presents the CPU runtime of all algorithms under different values of $M$. 
Although \texttt{R-WMMSE} and \texttt{R-MMMSE} do not demonstrate an advantage at $M = 128$, where the total number of users across all cells equals $M$, they show a notable advantage as $M$ grows, significantly outperforming \texttt{WMMSE} in computational time due to the $\mathcal{O}(M^3)$ complexity of \texttt{WMMSE}. 
They also surpass the hybrid BCD methods \texttt{A-WMMSE} and \texttt{A-MMMSE}, as first-order methods still incur $\mathcal{O}(M^2)$ complexity with respect to the number of BS antennas. 
The mixed weighted-unweighted warm-start method exhibits a clear advantage in large-antenna scenarios, making \texttt{R-MMMSE} the algorithm with the shortest CPU runtime for $M \in \{256, 512, 1024\}$.
As shown in Fig.~\ref{fig:t_ite}, the mixed weighted-unweighted warm-start method significantly reduces the number of iterations required by \texttt{WMMSE}, \texttt{R-WMMSE}, and \texttt{A-WMMSE}, indicating that our proposed warm-start strategy improves convergence in large-$M$ scenarios. 
Specifically, at $M = 1024$, \texttt{R-MMMSE} achieves a CPU runtime of 5.11 seconds, compared with 8.63 seconds for \texttt{R-WMMSE} and 489.17 seconds for \texttt{WMMSE}, corresponding to speedups of 40.8\% and 98.9\%, respectively.
Although first-order methods do not exhibit a significant advantage in this setting, their inherent suitability for parallel computation makes them promising for large-$M$ scenarios on GPUs, a point we will analyze further in Section~\ref{sec:gpu_experiment}.
}

{
To better analyze the acceleration effects of the hybrid BCD method and the mixed weighted-unweighted warm-start method, Fig.~\ref{fig:t_rduc} shows the percentage reduction in CPU runtime of \texttt{MMMSE}, \texttt{A-WMMSE}, and \texttt{A-MMMSE} relative to \texttt{WMMSE}. 
The results for \texttt{R-WMMSE} and \texttt{R-MMMSE} are omitted for readability, as their runtimes are too short to be meaningfully compared on the same scale. 
It can be observed that the improvements of both \texttt{MMMSE} (which adds only the warm-start method) and \texttt{A-WMMSE} (which adds only the hybrid BCD method) over \texttt{WMMSE} become more pronounced as $M$ increases. 
Furthermore, the acceleration achieved by \texttt{A-MMMSE}, which combines both techniques, also improves with $M$.
Although neither technique provides the same level of improvement as the low-dimensional subspace approach in \texttt{R-WMMSE} at large $M$, both still yield significant gains over \texttt{WMMSE}.
Specifically, at $M = 1024$, \texttt{A-MMMSE} achieves a 91.9\% reduction in CPU runtime compared to \texttt{WMMSE}.
}

\subsubsection{{Performance with Different $I$}}
Next, we evaluate algorithm performance with different user counts $I$, while fixing the number of cells $K = 4$, the number of BS antennas at $M = 256$, the number of data streams at $d = 4$, and the SNR at 10 dB, as shown in Fig.~\ref{fig:i_performance}.
Fig.~\ref{fig:i_sum_rate} confirms that all algorithms converge to almost the same WSR value under the given stopping conditions, validating their numerical consistency.
{In terms of computational time (Fig.~\ref{fig:i_cpu_time}), the hybrid BCD methods \texttt{A-WMMSE} and \texttt{A-MMMSE}, due to their linear complexity in the number of users, are significantly faster than the traditional BCM methods \texttt{WMMSE}, \texttt{MMMSE}, \texttt{R-WMMSE} and \texttt{R-MMMSE}. 
Since \texttt{R-WMMSE} eliminates the $\mathcal{O}(M^3)$ complexity but introduces cubic complexity in the total number of users across all cells, its CPU runtime exceeds that of \texttt{WMMSE} when the number of users per cell $I$ reaches 16. 
Our proposed mixed weighted-unweighted warm-start method brings notable improvements to \texttt{WMMSE}, \texttt{R-WMMSE}, and \texttt{A-WMMSE}. 
In the case of \texttt{R-WMMSE}, the warm-start strategy yields \texttt{R-MMMSE}, which outperforms \texttt{WMMSE} in terms of runtime.
However, it also suffers from the limitation that the low-dimensional subspace structure is not well suited for scenarios with a large number of users; consequently, when $I=16$, it is slower than \texttt{MMMSE}.
}

{
Fig.~\ref{fig:i_reduc} shows the percentage reduction in CPU runtime achieved by augmenting \texttt{WMMSE} with the hybrid BCD method and the mixed weighted-unweighted warm-start method. 
When the number of users is small, the mixed weighted-unweighted warm-start method provides greater improvements to \texttt{WMMSE}, which is reasonable because, when the number of users is smaller than the number of data streams, the savings in computational cost for $\mathbf{W}$ and the rate improvement dominate. 
As the number of users increases, the improvement in CPU runtime from using only the warm-start method (\texttt{MMMSE}) diminishes, whereas the hybrid BCD method becomes more advantageous. 
Specifically, the improvement of \texttt{A-WMMSE} over \texttt{WMMSE} exhibits an overall increasing trend with the number of users. 
Benefiting from both techniques, \texttt{A-MMMSE} achieves the best runtime when the number of users is large, attaining up to an 85.3\% reduction in CPU runtime compared to \texttt{WMMSE}.}

\begin{figure}
    \centering
    \includegraphics[width=1.0\linewidth]{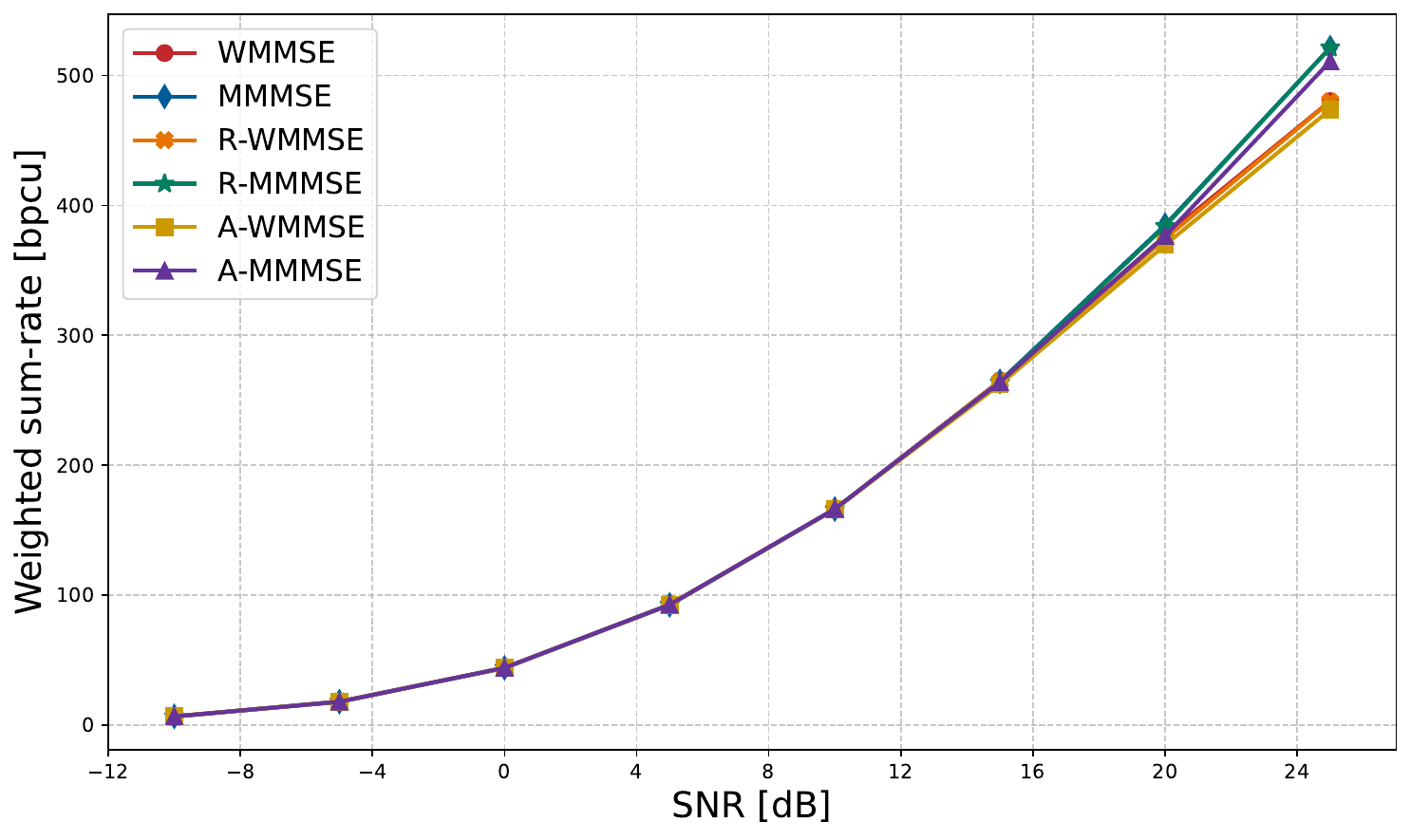}
    \caption{{WSR performance under varying SNR values, with $K=4$, $M=128$, $N=4$, $I=8$, and $d=4$.}}
    \label{fig:snr_sum_rate}
\end{figure}

\subsubsection{{Performance with Different SNR}}
{
Finally, we test all algorithms under different SNR values (Fig.~\ref{fig:snr_sum_rate}) with $K = 4$, $M = 128$, $I = 8$, and $d = 4$.
Since SNR variations do not affect algorithm complexity, we focus solely on whether our algorithms can converge to the same WSR value as \texttt{WMMSE} and \texttt{R-WMMSE} under varying noise levels. 
Below 15 dB, all algorithms converge to almost the same WSR value, and the WSR increases monotonically with SNR. 
When the SNR further increases to 25 dB, all \texttt{MMMSE}-based algorithms achieve higher WSR values than their \texttt{WMMSE}-based counterparts. 
This indicates that the mixed weighted-unweighted warm-start method provides a better initial point under high SNR conditions, enabling the second-stage WMMSE method to converge to a higher WSR value.}

\begin{figure}
    \centering
    \subfigure[{$K=2$, $I=16$, $N=4$, $d=4$, SNR=10dB.}]{\includegraphics[width=0.48\textwidth]{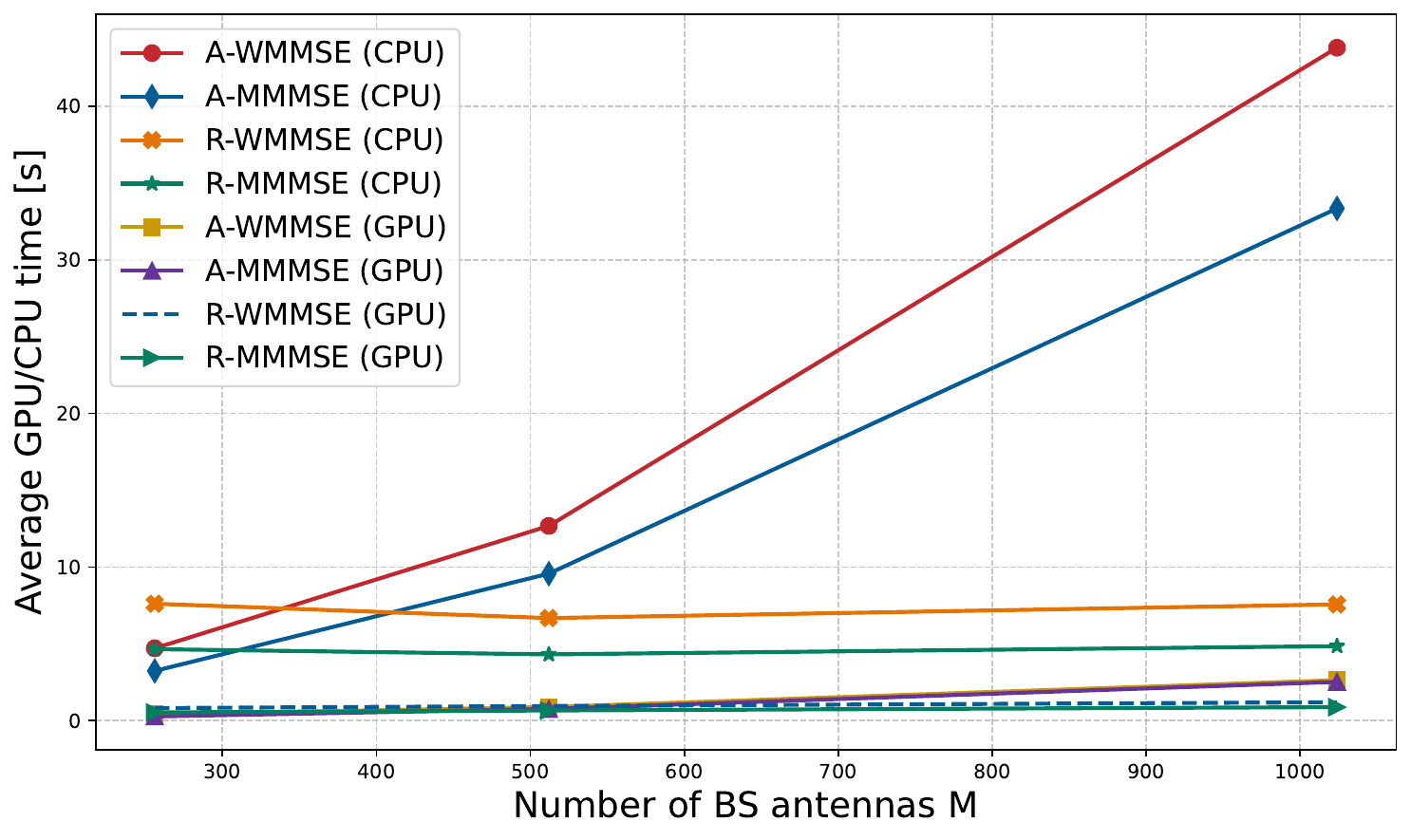} \label{fig:t_gpu_time_1}}
    \subfigure[{$K=2$, $I=32$, $N=4$, $d=4$, SNR=10dB.}]{\includegraphics[width=0.48\textwidth]{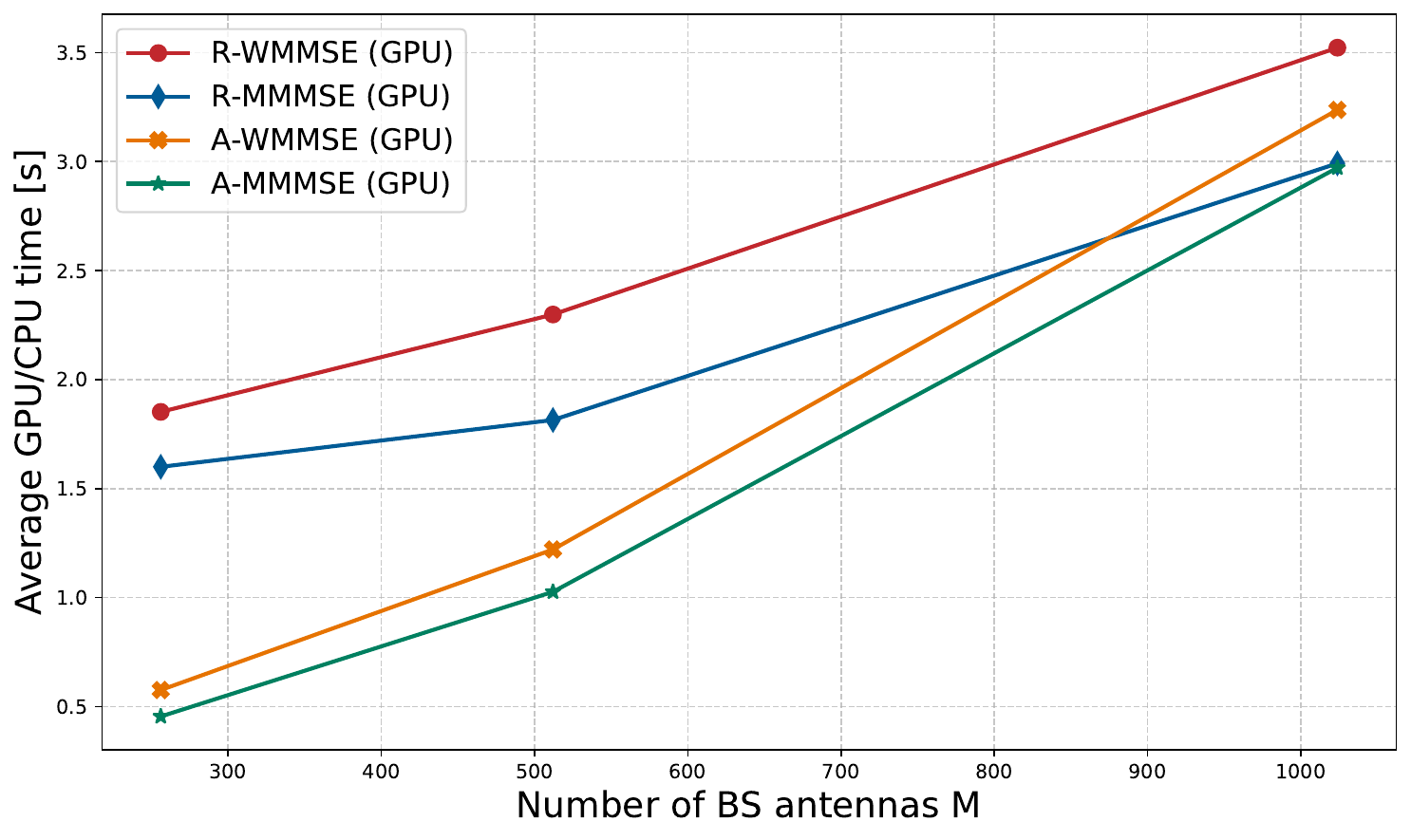} \label{fig:t_gpu_time_2}}
    \caption{{Average GPU/CPU runtime of \texttt{R-WMMSE} and \texttt{A-MMMSE} under varying numbers of BS antennas.}}
\label{fig:t_gpu_time}
\end{figure}

\subsection{Execution Time on GPUs}
\label{sec:gpu_experiment}

{
The hybrid BCD-based methods, \texttt{A-WMMSE} and \texttt{A-MMMSE}, are well suited for GPU acceleration due to the introduction of first-order techniques. 
As shown in Section~\ref{sec:cpu_time_sec}, when the number of BS antennas $M$ is large, \texttt{A-WMMSE} and \texttt{A-MMMSE} are slower than \texttt{R-WMMSE} and \texttt{R-MMMSE} (Fig.~\ref{fig:t_cpu_time}) because the hybrid BCD approach still incurs $\mathcal{O}(M^2)$ complexity when updating $\mathbf{V}$. 
To demonstrate the advantage of GPU acceleration for hybrid BCD methods, we implemented GPU versions of \texttt{R-WMMSE}, \texttt{R-MMMSE}, \texttt{A-WMMSE}, and \texttt{A-MMMSE} and compared them with their CPU counterparts. The results are presented in Fig.~\ref{fig:t_gpu_time}.

We first adopt the same configuration as in Fig.~\ref{fig:t_cpu_time}. 
It can be observed that when all algorithms are deployed on a GPU, their runtimes improve significantly, and the performance gap between \texttt{A-WMMSE}/\texttt{A-MMMSE} and \texttt{R-WMMSE}/\texttt{R-MMMSE} narrows. 
However, at $M = 1024$, the GPU runtimes of \texttt{A-WMMSE} and \texttt{A-MMMSE} remain slower than those of \texttt{R-WMMSE} and \texttt{R-MMMSE} due to the linear complexity of the latter in $M$.

We further increase the number of users to $I = 32$ (Fig.~\ref{fig:t_gpu_time_2}). 
To more clearly illustrate the runtime differences among \texttt{R-WMMSE}, \texttt{R-MMMSE}, \texttt{A-WMMSE}, and \texttt{A-MMMSE} on the GPU, we omit the CPU versions. 
Under this setting, \texttt{A-MMMSE} achieves the fastest GPU runtime for all values of $M$. 
\texttt{A-WMMSE} also outperforms \texttt{R-WMMSE} and is only slightly slower than \texttt{R-MMMSE} when $M=1024$. 
These results demonstrate that hybrid BCD methods are particularly suitable for GPU implementation, where they can fully exploit their parallelism advantages.
}

\begin{figure}
    \centering
    \subfigure[{Convergence in iterations.}]{\includegraphics[width=0.48\textwidth]{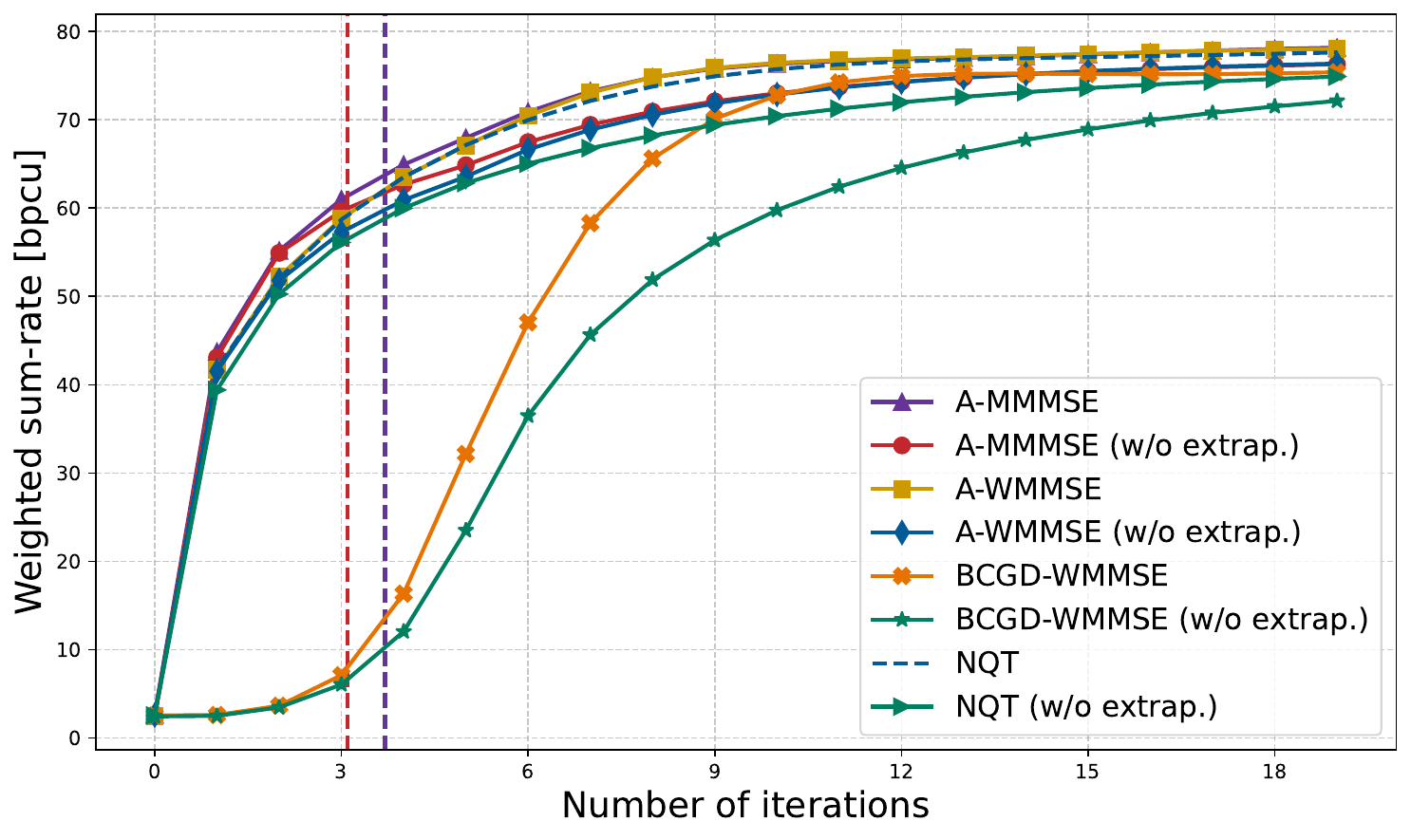} \label{fig:ite_ablation}}
    \subfigure[{Convergence in time.}]{\includegraphics[width=0.48\textwidth]{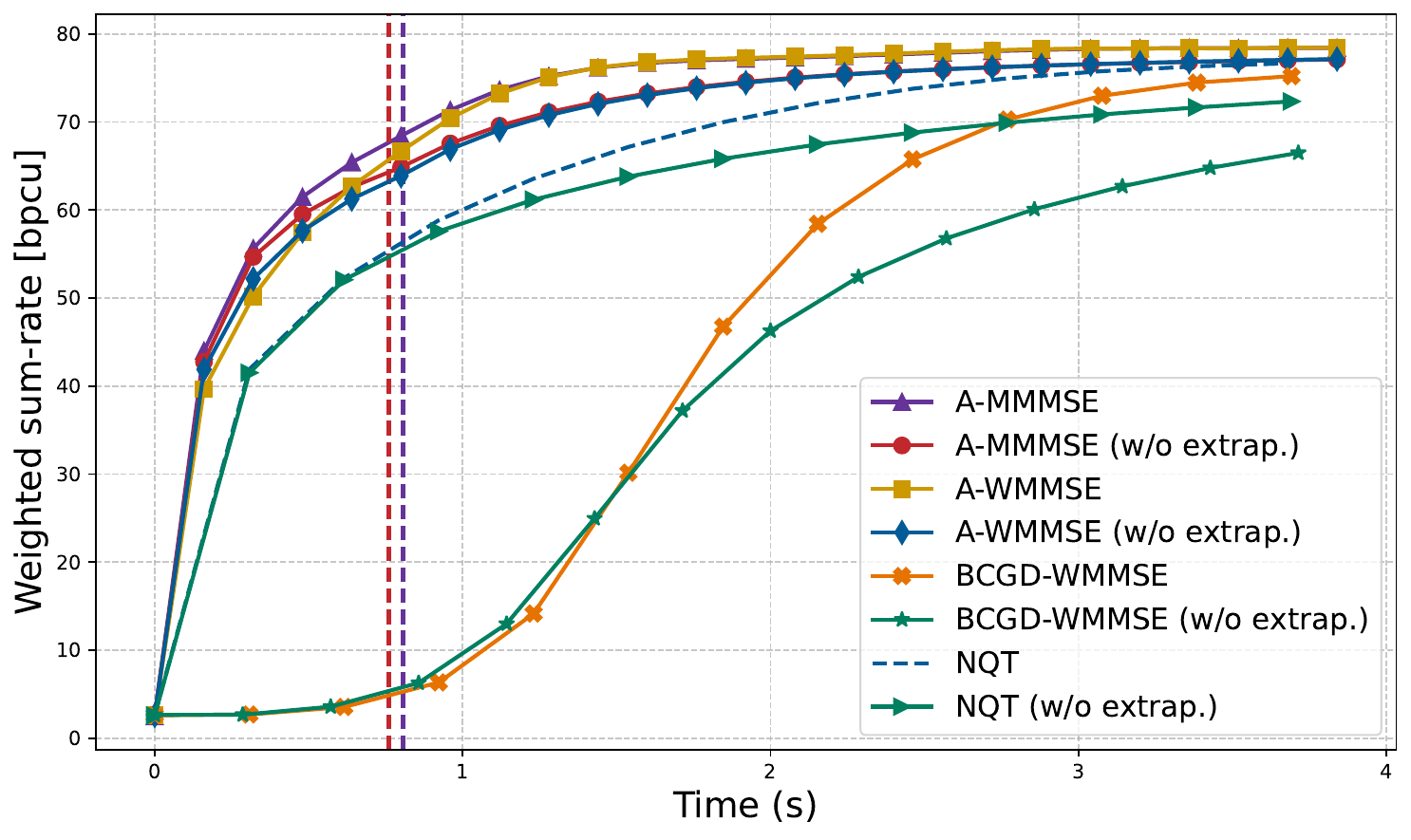} \label{fig:time_ablation}}
    \caption{{Convergence analysis of gradient-based methods with $K = 2$, $I = 12$, $M = 256$, $N=4$, $\text{SNR} = 10$ dB, and $d = 1$.}}
\label{fig:ablation}
\end{figure}
\subsection{{Further Analysis of Hybrid BCD and Warm-Start}}
\label{sec:gd_method}
{
\subsubsection{Ablation Study of Hybrid BCD}
In this section, we further analyze the advantages of the hybrid BCD method and the extrapolation technique. 
We evaluate the convergence of several gradient-based methods from two perspectives: a fixed number of iterations (20 iterations, Fig.~\ref{fig:ite_ablation}) and a fixed runtime (4 seconds, Fig.~\ref{fig:time_ablation}). 
The results are presented in Fig.~\ref{fig:ablation}. 
Methods whose names include \enquote{(w/o extrap.)} are those without Nesterov's extrapolation. 
We additionally include a full BCGD method, which updates all variables $\mathbf{U}$, $\mathbf{W}$, and $\mathbf{V}$ using a single gradient descent step; this method is denoted as \texttt{BCGD-WMMSE}. 
Since the convergence of WMMSE requires $\mathbf{W}$ to remain positive semidefinite, we update $\mathbf{W}$ using the Schulz iteration as adopted in~\cite{pellaco2023matrix}:
\begin{equation}
\mathbf{W}_{i_k}^{t+1} = \mathbf{W}_{i_k}^{t} \left(2\mathbf{I} - \mathbf{E}_{i_k}^t \mathbf{W}_{i_k}^t\right),\quad\forall i_k\in\mathcal{I}
\end{equation}
where $\mathbf{E}_{i_k}$ is defined in~\eqref{eq:Ek}. 
In the experiments, we also include \texttt{NQT} and its variant without extrapolation. 
Thus, this experiment covers the major existing gradient-based methods for solving problem~\eqref{prob:ori_prob}. 
To enable a fair comparison with \texttt{NQT}, we consider only the single-data-stream case ($d = 1$) and set $K=2$, $I = 12$, $M = 256$, $N = 4$, and $\text{SNR} = 10$ dB.

It can be observed that \texttt{A-MMMSE}, \texttt{A-WMMSE}, and \texttt{NQT} with extrapolation achieve the highest WSR values under both the fixed-iteration and fixed-time settings. 
When extrapolation is removed, the convergence rates of these algorithms degrade significantly, such that they can only approach—but not reach—the same WSR values within the given iteration or time budgets. 
The full BCGD method, \texttt{BCGD-WMMSE}, exhibits very slow convergence. 
This is because BCGD provides only approximate solutions to subproblems, and using first-order methods to update the relatively inexpensive variables $\mathbf{U}$ and $\mathbf{W}$ substantially reduces the overall convergence rate. 
The version without extrapolation, \texttt{BCGD-WMMSE (w/o extrap.)}, converges even more slowly. 
The superior performance of \texttt{A-MMMSE} and \texttt{A-WMMSE}, which employ the hybrid BCD approach, highlights the advantage of our hybrid BCD strategy. 
Notably, \texttt{NQT} and \texttt{A-WMMSE} exhibit similar convergence rates, suggesting a strong correlation between the two methods. Although the advantage of the mixed weighted-unweighted warm-start method is modest in the \(d = 1\) case, it still enables \texttt{A-MMMSE} to achieve the fastest convergence rate, thereby demonstrating the effectiveness of our warm-start strategy.

\begin{figure}
    \centering
    \subfigure[{$K=4$, $I=8$, $M=128$, $N=4$, $d=4$, SNR=-10dB.}]{\includegraphics[width=0.48\textwidth]{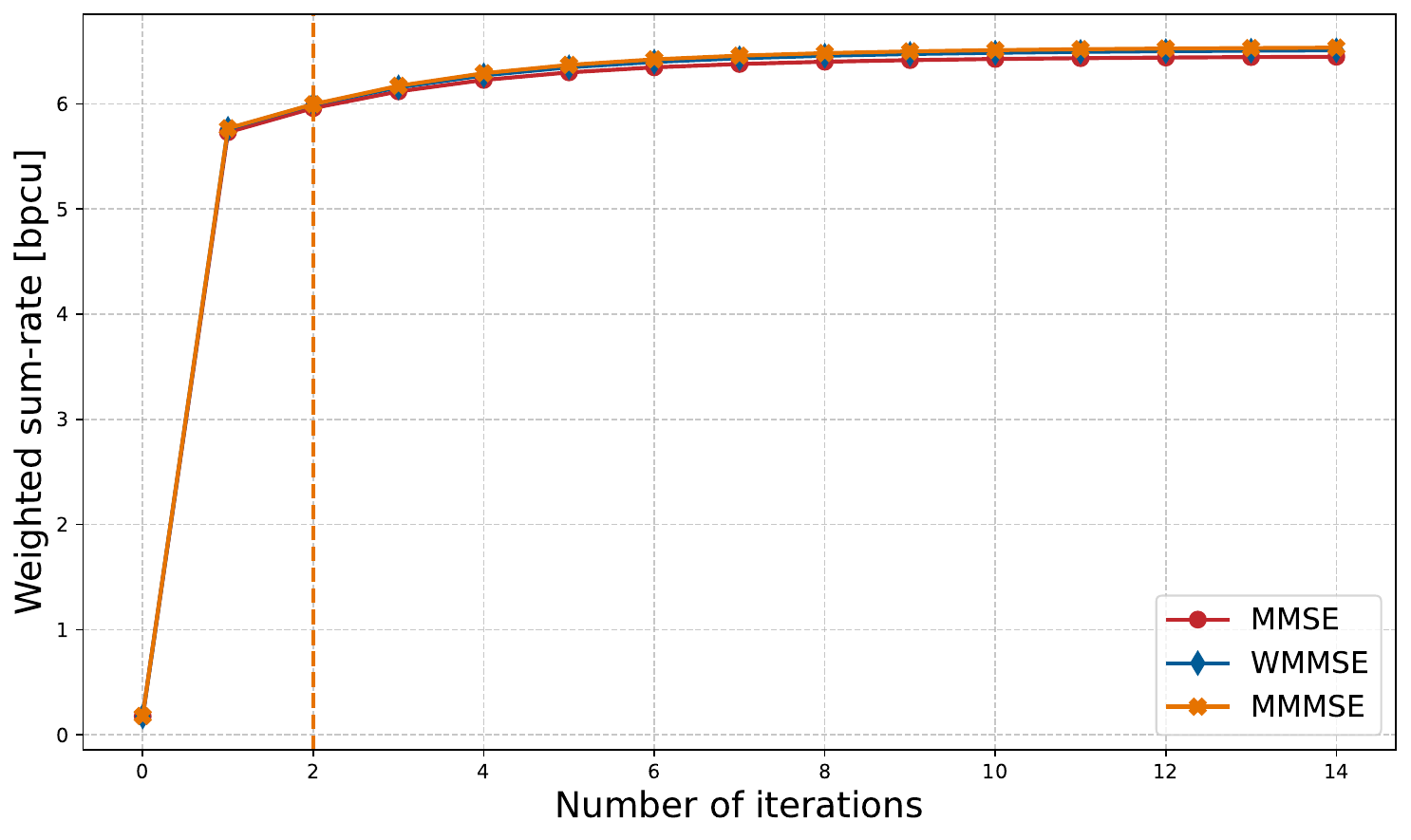} \label{fig:ws_convg_1}}
    \subfigure[{$K=4$, $I=16$, $M=256$, $N=4$, $d=4$, SNR=10dB.}]{\includegraphics[width=0.48\textwidth]{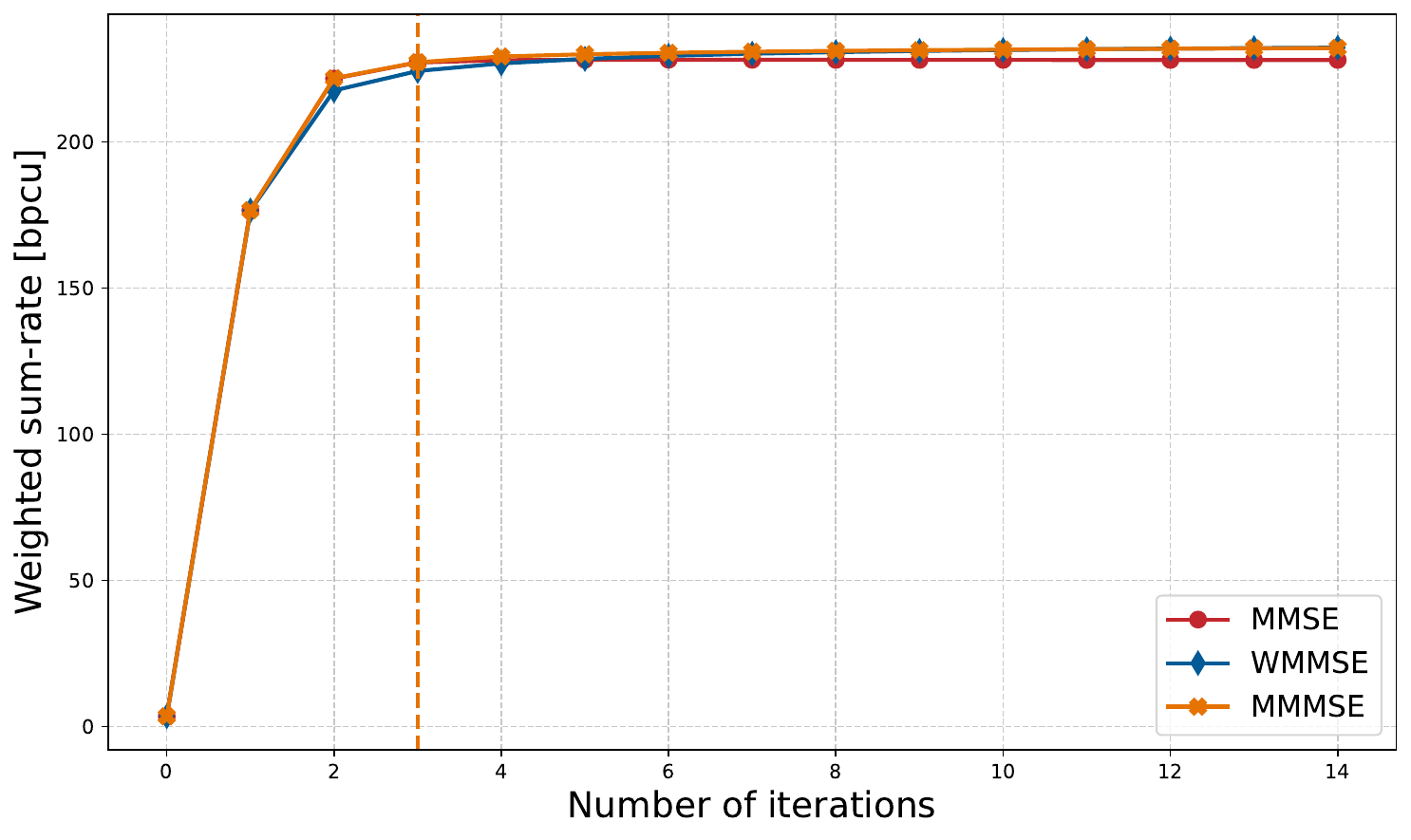} \label{fig:ws_convg_2}}
    \subfigure[{$K=4$, $I=24$, $M=512$, $N=4$, $d=4$, SNR=20dB.}]{\includegraphics[width=0.48\textwidth]{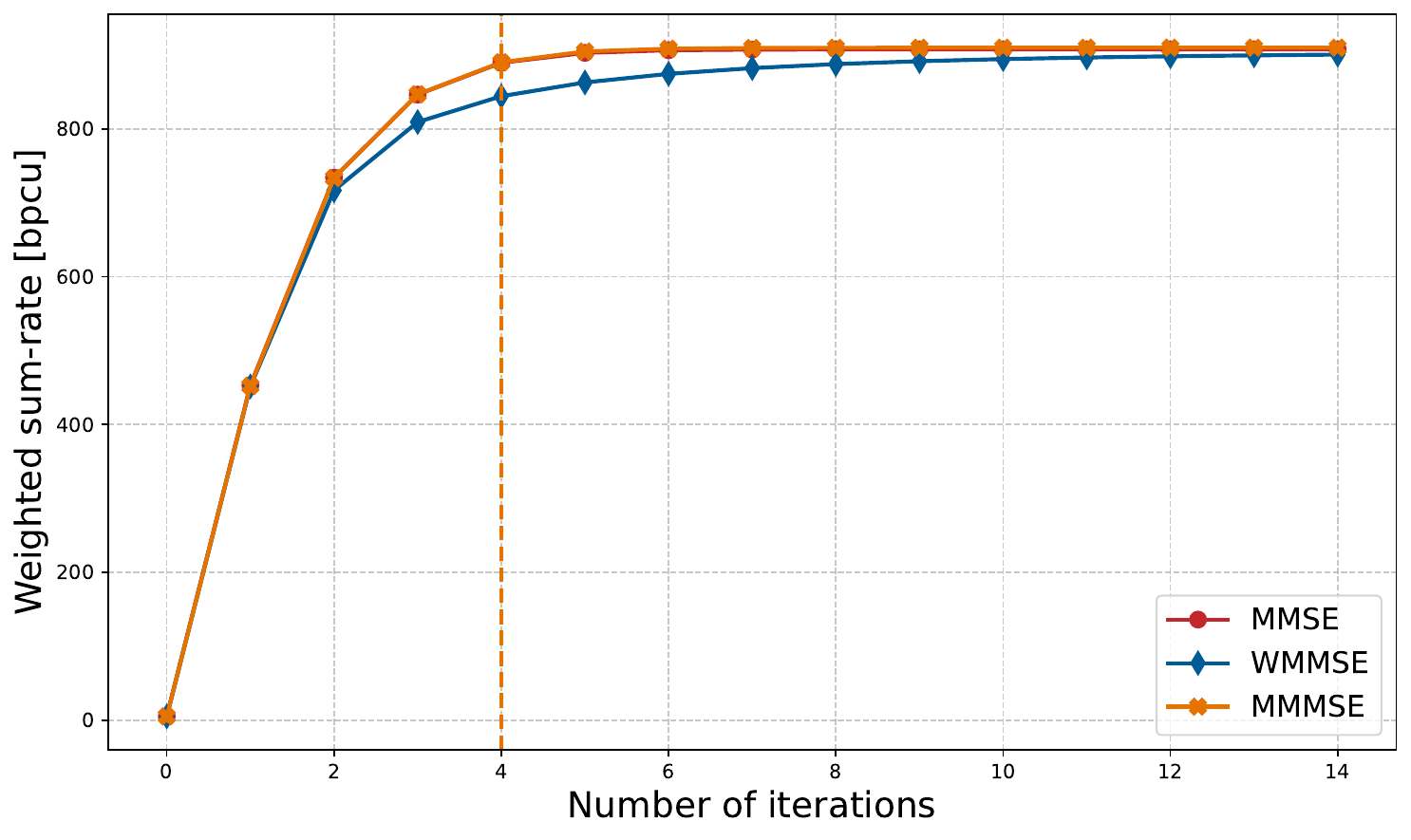} \label{fig:ws_convg_3}}
    \caption{{Convergence performance of \texttt{MMSE}, \texttt{WMMSE} and \texttt{MMMSE}.}}
\label{fig:ws_convg}
\end{figure}

\subsubsection{Convergence and Conditioning Analysis of Warm-Start}
To analyze the improvement in convergence rate brought by the mixed weighted-unweighted warm-start method, we configured three different scenarios with varying system scales and SNR levels. 
We fixed the number of iterations and separately analyzed the convergence performance of Stage I (\texttt{MMSE}), Stage II (\texttt{WMMSE}), and their combination (\texttt{MMMSE}). 
As shown in Fig.~\ref{fig:ws_convg}, \texttt{MMSE} achieves a faster convergence rate than \texttt{WMMSE} in the early iterations, particularly in larger-scale scenarios (Figs.~\ref{fig:ws_convg_2} and~\ref{fig:ws_convg_3}). 
However, because the problem formulation of \texttt{MMSE} is not equivalent to the original WSR maximization problem~\eqref{prob:ori_prob}, the final WSR value it converges to deviates from that of \texttt{WMMSE}. 
By combining both stages, \texttt{MMMSE} inherits the advantages of each: it achieves fast convergence in the early iterations while still guaranteeing convergence to a stationary point of~\eqref{prob:ori_prob}.
We also observe a phenomenon similar to that in Fig.~\ref{fig:snr_sum_rate} under high SNR conditions. 
Specifically, \texttt{MMSE} achieves rapid convergence in the early iterations and attains a higher WSR value than \texttt{WMMSE}. 
\texttt{MMMSE} inherits the high-WSR initial point provided by \texttt{MMSE}, then switches to the \texttt{WMMSE} update rule and converges to a stationary point of problem~\eqref{prob:ori_prob}, thereby also achieving a higher final WSR value.

}

\section{Conclusions}
\label{sec:concl}
This paper has presented A-MMMSE, an enhanced algorithm for WSR maximization in downlink multi-cell MU-MIMO systems, extending the classical WMMSE approach. 
The proposed method integrates two key innovations: a computationally efficient precoding update scheme and a heuristic warm-start strategy grounded in sum MSE minimization. 
Theoretical analysis and experimental validation have confirmed that A-MMMSE converges to a stationary point of the WSR maximization problem while attaining almost the same final WSR performance as both WMMSE and its refined variant R-WMMSE. 
Owing to its reduced computational complexity and high parallelizability, the algorithm achieves superior efficiency on CPU and GPU platforms compared to existing methods. 
To further improve usability, future research will aim to develop self-adaptive mechanisms that eliminate the need for manual parameter tuning, building upon the complexity reduction achieved in this work.

{\appendices
\section{Proof of Lemma~\ref{lemma:lambda_ek}}
~\label{append:proof_lambda_ek}
From the definition of $\kappa$ in \eqref{eq:sigma_HH}, we know that {$\|\mathbf{H}_{i_k}\|_2 \leq \sqrt{\kappa}$} for all $i_k$. Thus, 
{
\begin{equation}
\left\|\mathbf{H}_{i_k} \mathbf{V}_{i_k}\right\|_{F}^{2} \leq \kappa\left\|\mathbf{V}_{i_k}\right\|_{F}^{2} \leq \kappa \sum_{i=1}^I\left\|\mathbf{V}_{i_k}\right\|_{F}^{2} \leq \kappa P_{\max} .
\end{equation}}
Given the definition of {$\mathbf{E}_{i_k}$}, it follows that
{
\begin{small}
\begin{equation}
\lambda_{\min}\left(\mathbf{E}_{i_k}\right) \geq \lambda_{\min}\left(\mathbf{E}_{i_k}-\sum_{(l,j)\neq(i,k)}\mathbf{U}_{i_k}^{H} \mathbf{H}_{i_kj} \mathbf{V}_{l_j} \mathbf{V}_{l_j}^{H} \mathbf{H}_{i_kj}^{H}\mathbf{U}_{i_k}\right).
\end{equation}
\end{small}}
For the purpose of this proof, we therefore consider {$\mathbf{E}_{i_k}$} without the interference term {$\sum_{(l,j) \neq (i,k)} \mathbf{U}_{i_k}^{H} \mathbf{H}_{i_kj} \mathbf{V}_{l_j} \mathbf{V}_{l_j}^{H} \mathbf{H}_{i_kj}^{H} \mathbf{U}_{i_k}$}.
Let us examine {$\mathbf{a}^{H} \mathbf{E}_{i_k} \mathbf{a}$} for an arbitrary vector \(\mathbf{a}\) satisfying \(\|\mathbf{a}\|_{2} = 1\), and define {$\mathbf{b} = \mathbf{U}_{i_k} \mathbf{a}$}. 
Then,
{
\begin{equation}
\mathbf{a}^{H} \mathbf{E}_{i_k} \mathbf{a} = \left\|\mathbf{a} - \mathbf{V}_{i_k}^{H} \mathbf{H}_{i_kk}^{H} \mathbf{b}\right\|_{2}^{2} + \sigma_{i_k}^2 \|\mathbf{b}\|_{2}^{2}.
\end{equation}}
If we assume that $\|\mathbf{b}\|_{2}^{2} \geq \frac{1}{ P_{\max}\kappa}$, it follows immediately that
\begin{equation}
\mathbf{a}^{H} {\mathbf{E}_{i_k}} \mathbf{a} \geq {\sigma_{i_k}^2} \|\mathbf{b}\|_{2}^{2} \geq \frac{{\sigma_{i_k}^{2}}}{P_{\max} \kappa} \geq \frac{{\sigma_{i_k}^{2}}}{P_{\max} \kappa + {\sigma_{i_k}^{2}}}.
\end{equation}
Alternatively, if \(\|\mathbf{b}\|_{2}^{2} \leq \frac{1}{P_{\max} \kappa}\), or equivalently \(\|\mathbf{b}\|_{2} \leq \frac{1}{\sqrt{P_{\max} \kappa}}\), then
\begin{equation}
\begin{aligned}
\left\|\mathbf{a} - {\mathbf{V}_{i_k}^{H}} {\mathbf{H}_{i_kk}^{H}} \mathbf{b}\right\|_{2} &\geq \|\mathbf{a}\|_{2} - \left\|{\mathbf{V}_{i_k}^{H}} {\mathbf{H}_{i_kk}^{H}} \mathbf{b}\right\|_{2} \\
&\geq 1 - \sqrt{P_{\max} \kappa} \, \|\mathbf{b}\|_{2} \geq 0,
\end{aligned}
\end{equation}
which implies that
\begin{equation}
    \begin{aligned}
        \left\|\mathbf{a}-{\mathbf{V}_{i_k}^{H}} {\mathbf{H}_{i_kk}^{H}} \mathbf{b}\right\|_{2}^{2} & \geq\left(1-\sqrt{P_{\max}\kappa}\|\mathbf{b}\|_{2}\right)^{2}\\
        & =1-2 \sqrt{P_{\max} \kappa}\|\mathbf{b}\|_{2}+P_{\max} \kappa\|\mathbf{b}\|_{2}^{2}.\\
    \end{aligned}
\end{equation}
Therefore,
\begin{equation}
\begin{aligned}
\mathbf{a}^{H} {\mathbf{E}_{i_k}} \mathbf{a} & =\left\|\mathbf{a}-{\mathbf{V}_{i_k}^{H}} {\mathbf{H}_{i_kk}^{H}} \mathbf{b}\right\|_{2}^{2}+{\sigma_{i_k}^2}\|\mathbf{b}\|_{2}^{2} \\
& \geq 1-2 \sqrt{P_{\max} \kappa}\|\mathbf{b}\|_{2}+P_{\max} \kappa\|\mathbf{b}\|_{2}^{2}+ {\sigma_{i_k}^{2}}\|\mathbf{b}\|_{2}^{2} .
\end{aligned}
\end{equation}
Minimizing the right-hand side over the scalar $\|\mathbf{b}\|_{2}$ yields
\begin{equation}
\mathbf{a}^{H} {\mathbf{E}_{i_k}} \mathbf{a} \geq \frac{{\sigma_{i_k}^{2}}}{P_{\max} \kappa+{\sigma_{i_k}^{2}}} .
\end{equation}
Hence, for any $\mathbf{a}$ where $\|\mathbf{a}\|_{2}=1$,
\begin{equation}
\sigma_{\min }\left({\mathbf{E}_{i_k}}\right)=\min _{\|\mathbf{a}\|=1} \mathbf{a}^{H} {\mathbf{E}_{i_k}} \mathbf{a} \geq \frac{{\sigma_{i_k}^{2}}}{P_{\max} \kappa+{\sigma_{i_k}^{2}}},
\end{equation}
and consequently,
\begin{equation}
\lambda_{\min}\left({\mathbf{E}_{i_k}}\right) = \sigma_{\min}\left({\mathbf{E}_{i_k}}\right) \geq \frac{{\sigma_{i_k}^{2}}}{P_{\max} \kappa + {\sigma_{i_k}^{2}}}.
\end{equation}

\section{Proof of Lemma~\ref{lemma:limit_point}}
\label{append:proof_limit_point}
To prove the statement, we demonstrate that the iterates $\{(\mathbf{U}^t, \mathbf{W}^t, \mathbf{V}^t)\}_t$ generated by the A-MMMSE algorithm are confined to a compact set. 
The constraint in problem~\eqref{prob:wmmse_prob} ensures that $\mathbf{V}$ remains within a compact set due to projection. 
{From equation~\eqref{eq:Ek} and~\eqref{eq:Uk}}, we observe that
\begin{equation}
\begin{aligned}
\|{\mathbf{U}_{i_k}}\|_F^2 &=  \text{Tr}\left({\mathbf{U}_{i_k}^{H}\mathbf{U}_{i_k}}\right)\leq \frac{\text{Tr}({\mathbf{E}_{i_k}})}{{\sigma_{i_k}^2}}\leq \frac{d}{{\sigma_{i_k}^2}}.
\end{aligned}
\end{equation}
{Since $\mathbf{W}_{i_k} = \mathbf{E}_{i_k}^{-1}$}, Lemma~\ref{lemma:lambda_ek} implies that {$\lambda_{\max}(\mathbf{W}_{i_k}) \leq \frac{P_{\max} \kappa + \sigma_{i_k}^2}{\sigma_{i_k}^2}$} for all $i_k$, and therefore
{
\begin{equation}
\left\|\mathbf{W}_{i_k}\right\|_{F}^{2} = \sum_{i=1}^{d} \lambda_i^2(\mathbf{W}_{i_k}) \leq d\frac{\left(P_{\max} \kappa + \sigma_{i_k}^{2}\right)^{2}}{\sigma_{i_k}^{4}}, \quad \forall i_k \in \mathcal{I}.
\end{equation}}Thus, the sequence $\{(\mathbf{U}^t, \mathbf{W}^t, \mathbf{V}^t)\}_t$ is confined to a compact set. 
By the {Bolzano–Weierstrass theorem~\cite{dym2004principles}} ({which states that every bounded sequence has a convergent subsequence}), there exists a convergent subsequence $\{(\mathbf{U}, \mathbf{W}, \mathbf{V})^{t_j}\}_{t_j}$ that approaches a limit point.

\section{Proof of Lemma~\ref{lemma:L_V}}
\label{append:proof_lemma_L_V}
Define \(f(\mathbf{V}) \triangleq f(\mathbf{U}, \mathbf{W}, \mathbf{V})\). 
We have
\begin{equation}
\left\|\nabla_{v} f\left({\dot{\mathbf{V}}_{i_k}}\right)-\nabla_{v} f\left({\ddot{\mathbf{V}}_{i_k}}\right)\right\|_{F} \leq\left\|{\mathbf{F}_{i_k}}\right\|_{2}\left\|{\dot{\mathbf{V}}_{i_k}}-{\ddot{\mathbf{V}}_{i_k}}\right\|_{F},
\end{equation}
where
{
\begin{equation}
    \mathbf{F}_{i_k} = \sum_{(l,j)} 2\alpha_{l_j}\mathbf{H}_{l_jk}^H\mathbf{U}_{l_j}\mathbf{W}_{l_j}\mathbf{U}_{l_j}^{H}\mathbf{H}_{l_jk}.
\end{equation}}
From the inequality
{
\begin{equation}
\lambda_{\max }\left(\mathbf{W}_{l_j} \mathbf{U}_{l_j}^{H} \mathbf{U}_{l_j}\right) \leq \frac{\lambda_{\max }\left(\mathbf{W}_{l_j} \mathbf{E}_{l_j}\right)}{\sigma_{l_j}^{2}} = \frac{1}{\sigma_{l_j}^{2}}, \quad \forall l_j\in \mathcal{I},
\end{equation}}
it follows that
\begin{equation}
\begin{aligned}
\|{\mathbf{F}_{i_k}}\|_2&\leq 2\bar{\alpha}{KI}\left\| {\mathbf{H}_{l_jk}^{H} \mathbf{U}_{l_j} \mathbf{W}_{l_j} \mathbf{U}_{l_j}^{H} \mathbf{H}_{l_jk}} \right\|_2 \\
& \leq 2\bar{\alpha}{KI} \left\| {\mathbf{H}_{l_jk}^{H}}\right\|_2\left\| {\mathbf{U}_{l_j} \mathbf{W}_{l_j} \mathbf{U}_{l_j}^{H}}\right\|_2\left\| {\mathbf{H}_{l_jk}} \right\|_2\\
&\leq \frac{2\bar{\alpha}{KI}\kappa}{{\sigma_{l_j}^2}}\leq \frac{2\bar{\alpha}{KI}\kappa}{\sigma^2},\quad \forall {i_k\in\mathcal{I}}.
\end{aligned}
\end{equation}

}


\bibliographystyle{IEEEtran}
\bibliography{ref}

\end{document}